\documentclass[3p]{elsarticle}

\AtBeginDocument{}
\usepackage{cmap}
\usepackage{bm}
\usepackage[T1]{fontenc}
\usepackage{float}
\usepackage{tikz}
\usepackage{amssymb,amsbsy,amsmath,amsfonts,amssymb,amscd}
\usepackage{verbatim}
\usepackage{wrapfig}
\usepackage{algorithmic}
\usepackage{multirow}
\usepackage{graphicx}
\usepackage{graphics}
\usepackage{stmaryrd}
\usepackage[small]{caption}
\usepackage{subcaption}
\usepackage{hyperref}
\usepackage{color}
\usepackage{amsthm}  
\usepackage{tabularx}
\usepackage{dcolumn}
\restylefloat{table}
\usepackage{booktabs}
\usepackage{xcolor,colortbl}
\usepackage{amssymb}
\usepackage{amsfonts}
\usepackage{graphicx}
\usepackage{parskip}
\newtheorem{definition}{Definition}
\newtheorem{prop}{Proposition}
\newtheorem{lemma}{Lemma}
\newtheorem{corol}{Corollary}
\newtheorem{theorem}{Theorem}
\newcommand{\norm}[1]{\left|\left|#1\right|\right|}

\hypersetup{
	colorlinks   = true, 
	urlcolor     = blue, 
	linkcolor    = black, 
	citecolor   = black 
}

\graphicspath{{figures/}}

\begin{document}

\begin{frontmatter}

\title{Entropy-stable hybridized discontinuous Galerkin methods for the compressible Euler and Navier-Stokes equations}

\author[add1,add2]{P. Fernandez\corref{1}}
\ead{pablof@mit.edu}
\author[add1,add2]{N.~C. Nguyen}
\ead{cuongng@mit.edu}
\author[add1,add2]{J. Peraire}
\ead{peraire@mit.edu}
\address[add1]{Department of Aeronautics and Astronautics, Massachusetts Institute of Technology, USA.}
\address[add2]{Center for Computational Engineering, Massachusetts Institute of Technology, USA.}
\cortext[1]{Corresponding author}

\begin{abstract}
In the spirit of making high-order discontinuous Galerkin (DG) methods more competitive, researchers have developed the hybridized DG methods, a class of discontinuous Galerkin methods that generalizes the Hybridizable DG (HDG), the Embedded DG (EDG) and the Interior Embedded DG (IEDG) methods. These methods are amenable to hybridization (static condensation) and thus to more computationally efficient implementations. Like other high-order DG methods, however, they may suffer from numerical stability issues in under-resolved fluid flow simulations. In this spirit, we introduce the hybridized DG methods for the compressible Euler and Navier-Stokes equations in entropy variables. Under a suitable choice of the numerical flux, the scheme can be shown to be entropy stable and satisfy the Second Law of Thermodynamics in an integral sense. The performance and robustness of the proposed family of schemes are illustrated through a series of steady and unsteady flow problems in subsonic, transonic, and supersonic regimes. 
The hybridized DG methods in entropy variables show the optimal accuracy order given by the polynomial approximation space, and are significantly superior to their counterparts in conservation variables in terms of stability and robustness, particularly for under-resolved and shock flows.
\end{abstract}


\end{frontmatter}

\section{Introduction}


Over the past few years, discontinuous Galerkin (DG) methods have emerged as a promising approach for fluid flow simulations. First, they allow for high-order discretizations on complex geometries and unstructured meshes; which is a critical feature to simulate transitional and turbulent flows over the complex three-dimensional geometries commonly encountered in industrial applications. Second, DG methods are well suited to emerging computing architectures, including graphics processing units (GPUs) and other many-core architectures, due to their high flop-to-communication ratio. The use of DG methods for large-eddy simulation (LES) of transitional and turbulent flows is being further encouraged by successful numerical predictions \cite{Beck:14,Fernandez:17a,Frere:15,Gassner:13,Murman:16,Renac:15,Uranga:11,Wiart:15}.

However, high-order DG methods remain computationally expensive and may suffer from numerical stability issues in under-resolved computations. In order to address the first limitation, researchers have recently developed the hybridized DG methods \cite{Fernandez:17a,Nguyen:15}, a class of discontinuous Galerkin methods that generalizes the Hybridizable DG (HDG) \cite{Cockburn:09a,Fidkowski:2015,Nguyen:12,Peraire:10,Woopen:2015}, the Embedded DG (EDG) \cite{Cockburn:09a,Cockburn:09b,Peraire:11} and the Interior Embedded DG (IEDG) \cite{Fernandez:16a} methods. In order to address the second issue, entropy-stable DG schemes have been proposed for the compressible Euler \cite{Barth:2006a,Hiltebrand:2014,Hou:2007,Jiang:1994} and Navier-Stokes equations \cite{Chandrashekar:2013,Gassner:BR1:2018,May:2015,Williams:17,Zakerzadeh:2017}. From a physical perspective, entropy stability implies the numerical solution satisfies the integral version of the Second Law of Thermodynamics in the computational domain, and this in turn allows for improved robustness in under-resolved computations \cite{Carpenter:2013,Diosady:2015,Fisher:2013,Fjordholm:2012b}.

In this paper, we 
devise entropy-stable hybridized DG methods for the compressible Euler and Navier-Stokes equations. To this end, we use entropy variables for the hybridized DG discretization; which allows us to derive an identity governing the evolution of total entropy in the numerical solution. The entropy stability of the scheme is then ensured by a proper choice of the numerical flux. 
The performance of the entropy-variable hybridized DG methods is illustrated through a series of steady and unsteady flows in subsonic, transonic, and supersonic regimes. Numerical results indicate the entropy-variable hybridized DG methods display optimal accuracy order, and are superior to their conservation-variable counterparts in terms of stability and robustness.

The remainder of the paper is organized as follows. In Section \ref{s:govEq}, we present the compressible Euler and Navier-Stokes equations, as well as a discussion on entropy pairs and symmetrization. The entropy-variable hybridized DG methods are introduced in Section \ref{s:evHybridizedDG}. Theoretical entropy stability results are discussed in Section \ref{s:stability}. Numerical examples for steady and unsteady flows are then presented in Section \ref{s:numExamples}. We conclude with some remarks in Section \ref{s:conclusions}.

\section{\label{s:govEq}Governing equations}

\subsection{The compressible Euler equations}

Let $t_f > 0$ be a final time and let $\Omega \subset \mathbb{R}^d, \, 1 \leq d \leq 3$ be an open, connected and bounded physical domain with Lipschitz boundary $\partial \Omega$. The unsteady compressible Euler equations in strong, conservation form are given by
\begin{subequations}
\label{euler}
\begin{alignat}{2}
\label{e:euler1}
\displaystyle \frac{\partial  \bm{u}}{\partial t} +  \nabla  \cdot  \bm{F}(\bm{u}) = 0 & , \qquad \mbox{in } \Omega \otimes (0, t_f) ,   \\
\label{e:euler2}
\bm{B}(\bm{u}) = 0 & , \qquad \mbox{on } \partial \Omega \otimes (0,t_f) , \\
\label{e:euler3}
\bm{u} - \bm{u}_0 = 0 & , \qquad \mbox{on } \Omega \otimes \{ 0 \} . 
\end{alignat}
\end{subequations}
Here, $\bm{u} = (\rho, \rho V_{j}, \rho E) \in X_u, \, \ j=1,...,d$ is the $m$-dimensional ($m = d + 2$) vector of conservation variables, $\bm{u}_0 \in X_u$ is an initial condition, $X_u \subset \mathbb{R}^m$ is the set of physical states (i.e.\ positive density and pressure), $\bm{B}(\bm{u})$ is a boundary operator, and $\bm{F}(\bm{u})$ are the inviscid fluxes of dimension $m \times d$,
\begin{equation}
\label{flux}
\bm{F}(\bm{u}) = \left( \begin{array}{c}
\rho V_j \\
\rho V_i V_j + \delta_{ij} p \\
 V_j (\rho E + p)
\end{array}
\right) , \qquad i , j = 1 , ... , d , 
\end{equation}
where $p$ denotes the thermodynamic pressure and $\delta_{ij}$ is the Kronecker delta. For a calorically perfect gas in thermodynamic equilibrium, $p = (\gamma - 1) \, \big( \rho E - \rho \, \norm{\bm{V}}^2 / 2 \big)$, where $\gamma = c_p / c_v > 1$ is the ratio of specific heats and in particular $\gamma \approx 1.4$ for air. $c_p$ and $c_v$ are the specific heats at constant pressure and volume, respectively. 
The steady-state compressible Euler equations are obtained by dropping Eq. \eqref{e:euler3} and the first term in Eq. \eqref{e:euler1}.

\subsection{The compressible Navier-Stokes equations}

The unsteady compressible Navier-Stokes equations in strong, conservation form are given by
\begin{subequations}
\label{e:NS}
\begin{alignat}{2}
\label{e:ns1}
\displaystyle \frac{\partial  \bm{u}}{\partial t} +  \nabla  \cdot  \bm{F}(\bm{u}) +  \nabla  \cdot  \bm{G}(\bm{u} , \nabla \bm{u}) = 0 & , \qquad \mbox{in } \Omega \otimes (0, t_f) ,   \\
\label{e:ns2}
\bm{B}(\bm{u} , \nabla \bm{u}) = 0 & , \qquad \mbox{on } \partial \Omega \otimes (0,t_f) , \\
\label{e:ns3}
\bm{u} - \bm{u}_0 = 0 & , \qquad \mbox{on } \Omega \otimes \{ 0 \} . 
\end{alignat}
\end{subequations}
where $\bm{G}(\bm{u} , \nabla \bm{u})$ are the viscous fluxes of dimension $m \times d$,
\begin{equation}
\label{e:viscFlux}
\bm{G}(\bm{u},\nabla \bm{u}) = - \left( \begin{array}{c}
0 \\
\tau_{ij}  \\
V_i \tau_{ij} - f_j
\end{array}
\right) , \qquad i , j = 1 , \dots , d . 
\end{equation}
For a Newtonian fluid with the Fourier's law of heat conduction, the viscous stress tensor and heat flux are given by
\begin{equation}
\label{e:closuresNS}
\tau_{ij} = \mu \, \bigg( \frac{\partial V_i}{\partial x_j}+\frac{\partial V_j}{\partial x_i} - \frac{2}{3}\frac{\partial V_k}{\partial x_k}\delta_{ij} \bigg) + \beta \, \frac{\partial V_k}{\partial x_k}\delta_{ij} ,  \qquad \qquad f_j = - \, \kappa \, \frac{\partial T}{\partial x_j} , 
\end{equation}
respectively, where summation over repeated indices is implied, and where $T$ denotes temperature, $\mu$ the dynamic (shear) viscosity, $\beta$ the bulk viscosity, $\kappa = c_p \, \mu / Pr$ the thermal conductivity, and $Pr$ the Prandtl number. In particular, $Pr \approx 0.71$ for air, and additionally $\beta = 0$ under the Stokes' hypothesis. 

Under these assumptions, the viscous fluxes are linear in $\nabla \bm{u}$ and can be written as
\begin{equation}
G_{ij}(\bm{u},\nabla \bm{u}) = - \big[ \mathcal{K}_{jk} (\bm{u}) \big]_{is} \ \frac{\partial u_s}{\partial x_k} , \quad i, \, s=1,\dots,m , \quad j, \, k=1,\dots,d , 
\end{equation}
where $\bm{\mathcal{K}}_{jk} (\bm{u}) \in \mathbb{R}^{m \times m}$ are symmetric positive semi-definite matrices \cite{Hughes:86}. 
The steady-state compressible Navier-Stokes equations are obtained by dropping Eq. \eqref{e:ns3} and the first term in Eq. \eqref{e:ns1}.


\subsection{Entropy pairs and symmetrization of the governing equations}

Nonlinear hyperbolic systems of conservation laws arising from physical systems, such as the compressible Euler equations, commonly admit a generalized entropy pair $(H(\bm{u}), \bm{\mathcal{F}}(\bm{u}))$ consisting of a convex generalized entropy function $H(\bm{u}) : \mathbb{R}^m \to \mathbb{R}$ and an entropy flux $\bm{\mathcal{F}}(\bm{u}) : \mathbb{R}^m \to \mathbb{R}^d$ that satisfies
\begin{equation}
\label{e:entropyFluxIdentity}
\frac{\partial \mathcal{F}_j}{\partial u_k} = \frac{\partial F_{ij}}{\partial u_k} \frac{\partial H}{\partial u_i}  , \quad i , \, k = 1 , ... , m  , \quad j=1 , ..., d . 
\end{equation}
Entropy pairs exist if and only if the hyperbolic system is symmetrized via the change of variables $\bm{v} (\bm{u}) = \partial H / \partial \bm{u}$ \cite{Godunov:61,Mock:80}, where $\bm{v}$ are referred to as the entropy variables. Such entropy pairs exist, among others, for the Euler equations and the magnetohydrodynamic (MHD) equations. An important property of entropy-symmetrized hyperbolic systems emerges when the inner product of the conservation law is taken with respect to the entropy variables, namely, the following identities hold for smooth solutions \cite{Barth:99}
\begin{equation}
\label{e:vTidentity}
\bm{v}^t \cdot \frac{\partial \, \bm{u}(\bm{v})}{\partial t} = \frac{\partial H}{\partial t} , \qquad \qquad \bm{v}^t \cdot \big( \nabla \cdot \bm{F}(\bm{v}) \big) = \nabla \cdot \bm{\mathcal{F}} (\bm{v}) . 
\end{equation}
These identities will be used for some of the proofs in \ref{s:appProof}. 
The following family of generalized entropy pairs for the Euler equations
\begin{equation}
H = -\rho \, g(s) , \qquad \qquad \bm{\mathcal{F}} = H \bm{V} , 
\end{equation}
was proposed by Harten \cite{Harten:83}, where $\bm{V}$ denotes the velocity vector, $s = \log(p / \rho^{\gamma}) - s_0$ is a non-dimensional thermodynamic entropy, $s_0$ is a baseline entropy level, and $g : \mathbb{R} \to \mathbb{R}$ is any smooth function such that $g' > 0$ and $g'' < g' / \gamma$. 
%
%
Among the entropy pairs in this family, only the subset of affine functions $g(s) = c_0 + c_1 s , \, c_1 > 0$ further symmetrizes the Navier-Stokes equations \cite{Hughes:86}. For this reason, we consider the following entropy pair in this paper
\begin{equation}
H = - \rho s , \qquad \qquad \bm{\mathcal{F}} = - \rho s \bm{V} , 
\end{equation}
which leads to the mapping
\begin{equation}
\bm{v} = \bm{v}(\bm{u}) = (\gamma-1) \left( \begin{array}{c}
\frac{-s}{\gamma-1} + \frac{\gamma+1}{\gamma-1} -\frac{\rho E}{p} \\
\frac{\rho \bm{V}}{p} \\
-\frac{\rho}{p}
\end{array}
\right) . 
\end{equation}
We shall denote the set of physical states in $\bm{v}$ space by $X_v$, i.e.\ $X_v = \bm{v}(X_u)$. The expressions for the inverse mapping $\bm{u} = \bm{u}(\bm{v})$ and the Jacobian matrices $\partial \bm{v} / \partial \bm{u}$ and $\partial \bm{u} / \partial \bm{v}$ are presented in \cite{Hughes:86}. 
We finally note that entropy-satisfying solutions of the Euler and Navier-Stokes equations satisfy
\begin{equation}
\label{e:entrIneq}
\frac{\partial H}{\partial t} + \nabla \cdot \bm{\mathcal{F}} \leq 0
\end{equation}
in the sense of distributions, where equality holds pointwise for smooth (classical) solutions of the Euler equations. Equation \eqref{e:entrIneq} follows from the entropy transport inequality (Second Law of Thermodynamics), and vice versa.

\section{\label{s:evHybridizedDG}The entropy-variable hybridized DG methods}

\subsection{Preliminaries and notation}

\subsubsection{\label{s:FEmesh}Finite element mesh}

We denote by $\mathcal{T}_h$ a collection of stationary, disjoint, non-singular, $p$-th degree curved elements $K$ that partition $\Omega$\footnote{Strictly speaking, the finite element mesh can only partition the problem domain if $\partial \Omega$ is piecewise $p$-th degree polynomial. For simplicity of exposition, and without loss of generality, we assume hereinafter that $\mathcal{T}_h$ actually partitions $\Omega$.}, and set $\partial \mathcal{T}_h := \{ \partial K : K \in \mathcal{T}_h \} $ to be the collection of the boundaries of the elements in $\mathcal{T}_h$. For an element $K$ of the collection $\mathcal{T}_h$, $F= \partial K \cap \partial \Omega$ is a boundary face if its $d-1$ Lebesgue measure is nonzero. For two elements $K^+$ and $K^-$ of $\mathcal{T}_h$, $F=\partial K^{+} \cap \partial K^{-}$ is the interior face between $K^+$ and $K^-$ if its $d-1$ Lebesgue measure is nonzero. We denote by $\mathcal{E}_h^I$ and $\mathcal{E}_h^B$ the set of interior and boundary faces, respectively, and we define $\mathcal{E}_h := \mathcal{E}_h^I \cup \mathcal{E}_h^B$ as the union of interior and boundary faces. Note that, by definition, $\partial \mathcal{T}_h$ and $\mathcal{E}_h$ are different. More precisely, an interior face is counted twice in $\partial \mathcal{T}_h$ but only once in $\mathcal{E}_h$, whereas a boundary face is counted once both in $\partial \mathcal{T}_h$ and $\mathcal{E}_h$.

\subsubsection{Finite element spaces}

Let $\mathcal{P}_{k}(D)$ denote the space of polynomials of degree at most $k$ on a domain $D \subset \mathbb{R}^n$, let $L^2(D)$ be the space of Lebesgue square-integrable functions on $D$, and $\mathcal{C}^0(D)$ the space of continuous functions on $D$. Also, let $\bm{\psi}^p_K$ denote the $p$-th degree parametric mapping from the reference element $K_{ref}$ to an element $K \in \mathcal{T}_h$ in the physical domain, and $\bm{\phi}^p_F$ be the $p$-th degree parametric mapping from the reference face $F_{ref}$ to a face $F \in \mathcal{E}_h$ in the physical domain. We then introduce the following discontinuous finite element spaces in $\mathcal{T}_h$,
\begin{subequations}
\begin{alignat}{2}
& \bm{\mathcal{Q}}_{h}^k &&= \big\{\bm{r} \in [L^2(\mathcal{T}_h)]^{m \times d} \ : \ (\bm{r} \circ \bm{\psi}_K^p )  |_K \in [\mathcal{P}_k(K_{ref})]^{m \times d} \ \ \forall K \in \mathcal{T}_h \big\} , \\
& \bm{\mathcal{V}}_{h}^k &&= \big\{\bm{w} \in [L^2(\mathcal{T}_h)]^m \ : \ (\bm{w} \circ \bm{\psi}_K^p )|_K \in [\mathcal{P}_k(K_{ref})]^m \ \ \forall K \in \mathcal{T}_h \big\} , 
\end{alignat}
\end{subequations}
and the following finite element spaces on the mesh skeleton $\mathcal{E}_{h}$,
\begin{subequations}
\begin{alignat}{2}
&\bm{\widehat{\mathcal{M}}}_{h}^k  &&= \big\{ \bm{\mu} \in [L^2(\mathcal{E}_h)]^m \ : \ (\bm{\mu} \circ \bm{\phi}^p_F) \in [\mathcal{P}^k(F_{ref})]^m \, \ \forall F \in \mathcal{E}_h \big\} , \\
&\bm{\widetilde{\mathcal{M}}}_{h}^k  &&= \big\{ \bm{\mu} \in [\mathcal{C}^0(\mathcal{E}_h)]^m \ : \ (\bm{\mu} \circ \bm{\phi}^p_F) \in [\mathcal{P}^k(F_{ref})]^m \, \ \forall F \in \mathcal{E}_h \big\} . 
\end{alignat}
\end{subequations}
Note that $\bm{\widehat{\mathcal{M}}}_{h}^k$ consists of functions which are discontinuous at the boundaries of the faces, whereas $\bm{\widetilde{\mathcal{M}}}_{h}^k$ consists of functions that are continuous at the boundaries of the faces. We also denote by $\bm{\mathcal{M}}_{h}^k$ a finite element space on $\mathcal{E}_{h}$ that satisfies $\bm{\widetilde{\mathcal{M}}}_{h}^k \subseteq \bm{\mathcal{M}}_{h}^k \subseteq \bm{\widehat{\mathcal{M}}}_{h}^k$. In particular, we define
$$\bm{\mathcal{M}}_{h}^k  = \big\{ \bm{\mu} \in [L^2(\mathcal{E}_h)]^m \ \\
\ : \ (\bm{\mu} \ \circ \ \bm{\phi}^p_F) \in [\mathcal{P}^k(F_{ref})]^m \, \ \forall F \in \mathcal{E}_h , \ \textnormal{and} \ \bm{\mu}|_{\mathcal{E}^{\rm E}_h} \in [\mathcal{C}^0(\mathcal{E}^{\rm E}_h)]^m    \big\} , $$
where $\mathcal{E}^{\rm E}_h$ is a subset of $\mathcal{E}_h$. Note that $\bm{\mathcal{M}}_{h}^k$ consists of functions which are continuous on $\mathcal{E}^{\rm E}_h$ and discontinuous on $\mathcal{E}^{\rm H}_h := \mathcal{E}_h \backslash \mathcal{E}^{\rm E}_h$. Furthermore, if $\mathcal{E}^{\rm E}_h = \emptyset$ then $\bm{\mathcal{M}}_{h}^k = \bm{\widehat{\mathcal{M}}}_{h}^k$, and if $\mathcal{E}^{\rm E}_h = \mathcal{E}_h$ then $\bm{\mathcal{M}}_{h}^k = \bm{\widetilde{\mathcal{M}}}_{h}^k$. Different choices of $\mathcal{E}^{\rm E}_h$ will lead to different discretization methods within the hybridized DG family.



It remains to define inner products associated with these finite element spaces. For functions $\bm{a}$ and $\bm{b}$ in $[L^2(D)]^m$, we denote $(\bm{a},\bm{b})_D = \int_{D} \bm{a} \cdot \bm{b}$  if $D$ is a domain in $\mathbb{R}^d$ and $\left\langle \bm{a},\bm{b}\right\rangle_D = \int_{D} \bm{a} \cdot \bm{b}$ if $D$ is a domain in $\mathbb{R}^{d-1}$. Likewise, for functions $\bm{A}$ and $\bm{B}$ in $[L^2(D)]^{m \times d}$, we denote $(\bm{A},\bm{B})_D = \int_{D} \mathrm{tr}(\bm{A}^t \bm{B})$  if $D$ is a domain in $\mathbb{R}^d$ and $\left\langle \bm{A},\bm{B}\right\rangle_D = \int_{D} \mathrm{tr}(\bm{A}^t \bm{B})$ if $D$ is a domain in $\mathbb{R}^{d-1}$, where $\mathrm{tr} \, ( \cdot) $ is the trace operator of a square matrix. We finally introduce the following inner products
\begin{equation}
\label{e:innerProducts}
(\bm{a},\bm{b})_{\mathcal{T}_h} = \sum_{K \in \mathcal{T}_h} (\bm{a},\bm{b})_K, \qquad (\bm{A},\bm{B})_{\mathcal{T}_h} = \sum_{K \in \mathcal{T}_h} (\bm{A},\bm{B})_K, \qquad \left\langle \bm{a},\bm{b}\right\rangle_{\partial \mathcal{T}_h} = \sum_{K \in \mathcal{T}_h} \left\langle \bm{a},\bm{b}\right\rangle_{\partial K} . 
\end{equation}

\subsection{\label{s:hDG_euler}The entropy-variable hybridized DG methods for the compressible Euler equations}

The entropy-variable hybridized DG discretization of the unsteady compressible Euler equations reads as follows: Find $\big( \bm{v}_h(t), \widehat{\bm{v}}_h(t) \big) \in \bm{\mathcal{V}}_h^k \otimes \bm{\mathcal{M}}_h^k$ such that
\begin{subequations}
\label{IEDG}
\begin{alignat}{2}
\label{e:hDG1}
\Big( \frac{\partial \, \bm{u}(\bm{v}_h)}{\partial t}, \bm{w} \Big)_{\mathcal{T}_h} - \Big( \bm{F}(\bm{v}_h) , \nabla \bm{w} \Big) _{\mathcal{T}_h}  +  \left\langle \widehat{\bm{f}}_h, \bm{w} \right\rangle_{\partial \mathcal{T}_h}  & = 0,  \\
\label{e:hDG2}
\left\langle \widehat{\bm{f}}_h, \bm{\mu} \right\rangle_{\partial \mathcal{T}_h \backslash \partial \Omega} + \left\langle \widehat{\bm{b}}_h(\widehat{\bm{v}}_h,\bm{v}_h;\bm{v}^{\partial \Omega}), \bm{\mu} \right\rangle_{\partial \Omega} & =  0 , \\
\intertext{for all $(\bm{w}, {\bm{\mu}}) \in \bm{\mathcal{V}}^k_h \otimes \bm{\mathcal{M}}_{h}^k$ and all $t \in (0,t_f)$, as well as}
\label{e:hDG3}
\big( \bm{v}_{h}|_{t=0} - \bm{v}(\bm{u}_0) , \bm{w} \big) _{\mathcal{T}_h} & =  0 , 
\end{alignat}
\end{subequations}
for all $\bm{w} \in \bm{\mathcal{V}}_h^k$. 
Here, $\bm{v}_h$ and $\widehat{\bm{v}}_h$ are the numerical approximations to $\bm{v}$ and $\bm{v}|_{\mathcal{E}_h}$, $\bm{F}(\bm{v}_h) = \bm{F}(\bm{u}(\bm{v}_h))$ denotes the inviscid flux in entropy variables, $\widehat{\bm{f}}_h$ is the inviscid numerical flux defined as
\begin{equation}
\label{numericalFlux}
\widehat{\bm{f}}_h = \widehat{\bm{f}}_h(\widehat{\bm{v}}_h , \bm{v}_h ) = \frac{1}{2} \big( \bm{F}(\widehat{\bm{v}}_h) + \bm{F}(\bm{v}_h) \big) \cdot \bm{n} + \frac{1}{2} \, \bm{\sigma}(\widehat{\bm{v}}_h , \bm{v}_h ; \bm{n}) \cdot ( \bm{v}_h - \widehat{\bm{v}}_h ) , 
\end{equation}
where $\bm{n}$ denotes the unit normal vector pointing outwards from the elements, and $\bm{\sigma} \in \mathbb{R}^{m \times m}$ is the so-called stabilization matrix. As discussed in Section \ref{s:stability}, the stabilization matrix plays an important role in the stability of the scheme. Also, $\widehat{\bm{b}}_h$ is the boundary condition term, whose precise definition depends on the type of boundary condition
, and $\bm{v}^{\partial \Omega}$ is a boundary state with support on $\partial \Omega$. The development of entropy-stable boundary conditions for hybridized DG methods is beyond the scope of this paper. 

Equation \eqref{e:hDG1} weakly imposes the Euler equations, Eq. \eqref{e:hDG2} weakly enforces the boundary conditions and the flux conservation across elements, and Eq. \eqref{e:hDG3} weakly imposes the initial condition. The entropy-variable hybridized DG discretization of the steady-state compressible Euler equations is obtained by dropping Eq. \eqref{e:hDG3} and the first term in Eq. \eqref{e:hDG1}. We note that, due to the discontinuous nature of $\bm{\mathcal{V}}^k_h$, Eq. \eqref{e:hDG1} can be used to locally (i.e.\ in an element-by-element fashion) eliminate $\bm{v}_h$ to obtain a weak formulation in terms of $\bm{\widehat{v}}_h$ only, and thus only the degrees of freedom of $\bm{\widehat{v}}_h$ are globally coupled \cite{Nguyen:15}. We finally note that, although conservation variables are not used as working variables in the discretization, the entropy-variable hybridized DG methods are $\bm{u}$-conservative; which follows by setting $\bm{w}$ and $\bm{\mu}$ to be constant functions in Equations \eqref{e:hDG1}$-$\eqref{e:hDG2}.


\subsection{\label{s:hDG_NS}The entropy-variable hybridized DG methods for the compressible Navier-Stokes equations}

The entropy-variable hybridized DG discretization of the unsteady compressible Navier-Stokes equations reads as follows: Find $\big( \bm{q}_h(t) , \bm{v}_h(t), \widehat{\bm{v}}_h(t) \big) \in \bm{\mathcal{Q}}_h^k \otimes \bm{\mathcal{V}}_h^k \otimes \bm{\mathcal{M}}_h^k$ such that
\begin{subequations}
\label{e:hDG_NS}
\begin{alignat}{2}
\label{e:hDG_NS0}
\big( \bm{q}_h, \bm{r} \big) _{\mathcal{T}_h} + \big( \bm{v}_h, \nabla \cdot \bm{r} \big)  _{\mathcal{T}_h} -  \big< \widehat{\bm{v}}_h, \bm{r} \cdot \bm{n} \big> _{\partial \mathcal{T}_h}  & =  0 , \\
\label{e:hDG_NS1}
\Big( \frac{\partial \, \bm{u}(\bm{v}_h)}{\partial t}, \bm{w} \Big)_{\mathcal{T}_h} - \Big( \bm{F}(\bm{v}_h) + \bm{G}(\bm{v}_h,\bm{q}_h) , \nabla \bm{w} \Big) _{\mathcal{T}_h}  +  \left\langle \widehat{\bm{f}}_h + \widehat{\bm{g}}_h, \bm{w} \right\rangle_{\partial \mathcal{T}_h}  & = 0,  \\
\label{e:hDG_NS2}
\left\langle \widehat{\bm{f}}_h + \widehat{\bm{g}}_h, \bm{\mu} \right\rangle_{\partial \mathcal{T}_h \backslash \partial \Omega} + \left\langle \widehat{\bm{b}}_h(\widehat{\bm{v}}_h,\bm{v}_h,\bm{q}_h;\bm{v}^{\partial \Omega}), \bm{\mu} \right\rangle_{\partial \Omega} & =  0 , \\
\intertext{for all $(\bm{r} , \bm{w} , {\bm{\mu}}) \in \bm{\mathcal{Q}}_h^k \otimes \bm{\mathcal{V}}^k_h \otimes \bm{\mathcal{M}}_{h}^k$ and all $t \in (0,t_f)$, as well as}
\label{e:hDG_NS3}
\big( \bm{v}_{h}|_{t=0} - \bm{v}(\bm{u}_0) , \bm{w} \big) _{\mathcal{T}_h} & =  0 , 
\end{alignat}
\end{subequations}
for all $\bm{w} \in \bm{\mathcal{V}}^k_h$. 
In addition to the nomenclature previously introduced, $\bm{q}_h$ is the numerical approximation to the gradient of the solution $\nabla \bm{v}$,
\begin{equation}
G_{ij}(\bm{v}_h,\bm{q}_h) = G_{ij}(\bm{u}(\bm{v}_h) , \nabla \bm{u}( \bm{v}_h , \bm{q}_h)) = - \big[ \mathcal{K}_{jk}(\bm{u}(\bm{v}_h)) \big]_{il} \  \frac{\partial u_l(\bm{v}_h)}{\partial v_s} \ q_{h,sk} , \quad i, \, l, \, s = 1 , \dots , m , \quad j , \, k = 1 , \dots , d , 
\end{equation}
is the viscous flux in entropy variables, and $\widehat{\bm{g}}_h$ is the viscous numerical flux. Inspired by the common choices in the context of conservation-variable hybridized DG methods \cite{Fernandez:17a,Nguyen:2009}, two options for $\widehat{\bm{g}}_h$ are
 \begin{subequations}
 \label{e:numericalFluxNS}
\begin{alignat}{2}
\label{e:numericalFluxNS_1}
\widehat{\bm{g}}_h & = \bm{G}(\bm{v}_h , \bm{q}_h) \cdot \bm{n} , \\
\label{e:numericalFluxNS_2}
\widehat{\bm{g}}_h & = \bm{G}(\widehat{\bm{v}}_h , \bm{q}_h) \cdot \bm{n} . 
\end{alignat}
\end{subequations}
Note again that the scheme is $\bm{u}$-conservative, and that Equations \eqref{e:hDG_NS1}$-$\eqref{e:hDG_NS2} can be used to locally eliminate both $\bm{q}_h$ and $\bm{v}_h$ to obtain a weak formulation in terms of $\bm{\widehat{v}}_h$ only. Hence, only the degrees of freedom of $\bm{\widehat{v}}_h$ are globally coupled. The entropy-variable hybridized DG discretization of the steady-state compressible Navier-Stokes equations is obtained by dropping Eq. \eqref{e:hDG_NS3} and the first term in Eq. \eqref{e:hDG_NS1}.

\subsection{Examples of schemes within the hybridized DG family}


Different choices of $\mathcal{E}^{\rm E}_h$ in the definition of the space $\bm{\mathcal{M}}_{h}^k$ lead to different schemes within the hybridized DG family. We present three interesting choices in this section. 
The first one is $\mathcal{E}^{\rm E}_h = \emptyset$ and yields $\bm{\mathcal{M}}_{h}^k=\bm{\widehat{\mathcal{M}}}_{h}^k$. This corresponds to the entropy-variable Hybridizable Discontinuous Galerkin (HDG) method. The second choice is $\mathcal{E}^{\rm E}_h = \mathcal{E}_h$ and implies $\bm{\mathcal{M}}_{h}^k=\bm{\widetilde{\mathcal{M}}}_{h}^k$. This corresponds to the entropy-variable Embedded Discontinuous Galerkin (EDG) method and makes the approximation space $\bm{\mathcal{M}}_{h}^k$ continuous over $\mathcal{E}_h$. Since $\widetilde{\bm{\mathcal{M}}}_{h}^k \subset \widehat{\bm{\mathcal{M}}}_{h}^k$, the EDG method has fewer globally coupled degrees of freedom that the HDG method. 
The third one is $\mathcal{E}^{\rm E}_h = \mathcal{E}^{I}_h$ and thus $\bm{\widetilde{\mathcal{M}}}_{h}^k \subset \bm{\mathcal{M}}_{h}^k \subset \bm{\widehat{\mathcal{M}}}_{h}^k$. The resulting approximation space consists of functions that are continuous everywhere but at the borders of the boundary faces. Therefore, the resulting method has an HDG flavor on the boundary faces and an EDG flavor on the interior faces, and is referred to as the entropy-variable Interior Embedded DG (IEDG) method. Figure \ref{figure1} illustrates the degrees of freedom for the HDG, IEDG and EDG methods in a four-element mesh. The three schemes differ from each other only in the degrees of freedom of the approximate trace $\widehat{\bm{v}}_h$.

\begin{figure}[htbp!]
\centering
\includegraphics[width=0.8\textwidth]{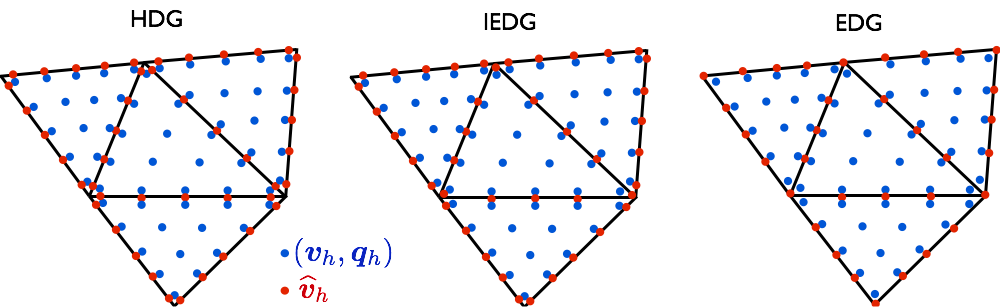}
\caption{\label{figure1} Illustration of the degrees of freedom for the HDG method, the IEDG method, and the EDG method. The blue nodes represent the degrees of freedom of the approximate solution $(\bm{v}_h, \bm{q}_h)$, while the red nodes represent the degrees of freedom of the approximate trace $\widehat{\bm{v}}_h$.}
\end{figure}

We note that the IEDG method enjoys advantages of both the HDG and the EDG methods. First, IEDG inherits the reduced number of global degrees of freedom and thus the computational efficiency of EDG, as discussed in Section \ref{s:comparisonSchemes}. In fact, the degrees of freedom of the approximate trace on $\mathcal{E}_h^B$ can be locally eliminated without affecting the sparsity pattern of the Jacobian matrix of the discretization, thus yielding an even smaller number of global degrees of freedom than in the EDG method. Second, the IEDG scheme enforces the boundary conditions as strongly as the HDG method, thus retaining the boundary condition robustness of HDG. These features make the IEDG method an excellent alternative to the HDG and EDG methods.

\subsection{\label{s:comparisonSchemes}Comparison with other DG methods}

We compare hybridized DG methods to ``standard'' DG methods, such as the Local DG (LDG) method \cite{Cockburn:LDG:98}, the Compact DG (CDG) method \cite{Peraire:CDG:08} or the BR2 method of Bassi and Rebay \cite{Bassi:BR:97}, in terms of the number of globally coupled degrees of freedom and the number of nonzero elements in the Jacobian matrix of the discretization. For polynomials of degree up to about five, the number of nonzero entries in the Jacobian matrix provides an indication of the computational cost of the scheme and, for many implicit time-integration implementations, also of the memory requirements. This is no longer a good cost metric for higher polynomial degrees since other operations, that are not accounted for by the number of nonzeros, start to dominate. The cases of triangular (2D) and tetrahedral (3D) meshes with polynomials of degree $k = 1, \dots , 5$ are considered. We assume that if $N_v$ is the number of mesh vertices, the number of triangles in 2D is about $2 N_v$ and the number of tetrahedra in 3D is about $6 N_v$. 
These assumptions are reasonable for large and well shaped meshes, and are consistent with those in \cite{Huerta:2013}. As will be discussed later, the accuracy order of the scheme is $k+1$.

For a hyperbolic system of conservation laws with $m$ components (e.g.\ $m = d+2$ for the Euler equations), the number of global degrees of freedom is given by
\begin{equation}
\label{e:dof}
\textnormal{DOF} = N_v \, m \, \alpha_{\textnormal{DOF}} , 
\end{equation}
whereas the number of nonzero entries in the Jacobian matrix is given by
\begin{equation}
\label{e:nnz}
\textnormal{NNZ} = N_v \, m^2 \, \alpha_{\textnormal{NNZ}} . 
\end{equation}
The coefficients $\alpha_{\textnormal{DOF}}$ and $\alpha_{\textnormal{NNZ}}$ are collected in Tables \ref{t:dofs} and \ref{t:nnz}, respectively. We note the coefficients $\alpha_{\textnormal{DOF}}$ and $\alpha_{\textnormal{NNZ}}$ for the IEDG method are bounded above by those of the EDG method, but cannot be determined exactly since they depend on the ratio of interior faces and boundary faces in the mesh. For second-order systems in space, such as the Navier-Stokes equations, the numbers in Tables \ref{t:dofs} and \ref{t:nnz} remain the same for hybridized DG methods but they increase for some instances of standard DG methods.

In all cases, EDG and IEDG provide a dramatic reduction in global degrees of freedom and number of nonzeros in the Jacobian matrix with respect to standard DG methods. For meshes in which the number of interior faces is much larger than the number of boundary faces, the coefficients for the IEDG method are only slightly smaller than those for the EDG method. For meshes in which the number of interior faces is about the number of boundary faces, the coefficients for the IEDG method are significantly smaller than those for the EDG method. 
Hence, for moderately high polynomial orders, hybridized DG methods are significantly superior to standard DG methods in terms of computational cost and memory requirements.

\begin{table}
\centering
\begin{tabular}{cccccccccccc}
\hline
 & \multicolumn{5}{c}{2D} & & \multicolumn{5}{c}{3D} \\
\hline
Degree & $k=1$ & $k=2$ & $k=3$ & $k=4$ & $k=5$ &  & $k=1$ & $k=2$ & $k=3$ & $k=4$ & $k=5$ \\
\hline
DG & $6$ & $12$ & $20$ & $30$ & $42$ &  & $24$ & $60$ & $120$ & $210$ & $336$ \\
HDG & $6$ & $9$ & $12$ & $15$ & $18$ &  & $36$ & $72$ & $120$ & $180$ & $252$ \\
EDG & $1$ & $4$ & $7$ & $10$ & $13$ &  & $1$ & $8$ & $27$ & $58$ & $101$ \\
IEDG & $< 1$ & $< 4$ & $< 7$ & $< 10$ & $< 13$ & & $< 1$ & $< 8$ & $< 27$ & $< 58$ & $< 101$\\
\hline
\end{tabular}
\caption{\label{t:dofs} Values of the coefficient $\alpha_{\textnormal{DOF}}$ for triangular and tetrahedral meshes as a function of the approximating polynomial order $k$ and the discretization algorithm. This coefficient can be used in Eq. \eqref{e:dof} to determine the number of global degrees of freedom in the problem.}
\end{table}

\begin{table}
\centering
\begin{tabular}{cccccccccccc}
\hline
 & \multicolumn{5}{c}{2D} & & \multicolumn{5}{c}{3D} \\
\hline
Degree & $k=1$ & $k=2$ & $k=3$ & $k=4$ & $k=5$ &  & $k=1$ & $k=2$ & $k=3$ & $k=4$ & $k=5$ \\
\hline
DG & $72$ & $288$ & $800$ & $1,800$ & $3,528$ &  & $480$ & $3,000$ & $12,000$ & $36,750$ & $94,080$ \\
HDG & $60$ & $135$ & $240$ & $375$ & $540$ &  & $756$ & $3,024$ & $8,400$ & $18,900$ & $37,044$ \\
EDG & $7$ & $46$ & $115$ & $214$ & $343$ &  & $15$ & $230$ & $1,311$ & $4,410$ & $11,183$ \\
IEDG & $< 7$ & $< 46$ & $< 115$ & $< 214$ & $< 343$ & & $< 15$ & $< 230$ & $< 1,311$ & $< 4,410$ & $< 11,183$ \\
\hline
\end{tabular}
\caption{\label{t:nnz} Values of the coefficient $\alpha_{\textnormal{NNZ}}$ for triangular and tetrahedral meshes as a function of the approximating polynomial order $k$ and the discretization algorithm. This coefficient can be used in Eq. \eqref{e:nnz} to determine the number of nonzero entries in the Jacobian matrix.}
\end{table}




\section{\label{s:stability}Stability properties}



\subsection{\label{s:notationStab}Preliminaries and notation}

Let $\bm{A}_n = \big( \partial \bm{F} / \partial \bm{u} \big) \cdot \bm{n}$ denote the Jacobian of the inviscid flux along the direction $\bm{n}$ with respect to the conservation variables, let $\tilde{\bm{A}}_0 = \partial \bm{u} / \partial \bm{v}$ be the (symmetric positive definite) Riemannian metric tensor of the change of variable $\bm{u} = \bm{u}(\bm{v})$, and let $\tilde{\bm{v}}(\theta ; \bm{v}_1 , \bm{v}_2) \in \mathcal{C}^{\infty}( [0,1] ; X_v)$ be such that $\tilde{\bm{v}}(0) = \bm{v}_1$, $\tilde{\bm{v}}(1) = \bm{v}_2$ and $\tilde{\bm{v}} \in X_v \ \forall \theta \in [0,1]$. Since $X_v$ is an open connected subset of $\mathbb{R}^m$, infinitely many such paths exist provided $\bm{v}_1 , \bm{v}_2 \in X_v$. The results hereinafter hold for any choice of path. Also, let us define $\tilde{\bm{A}}_n = \bm{A}_n \tilde{\bm{A}}_0$, and let $\big| \tilde{\bm{A}}_n \big|_{\tilde{\bm{A}}_0} = \tilde{\bm{A}}_0 \big| \tilde{\bm{A}}_0^{-1} \tilde{\bm{A}}_n \big| \equiv | \bm{A}_n | \tilde{\bm{A}}_0$ be the generalized absolute value operator with respect to the metric tensor $\tilde{\bm{A}}_0$. We finally introduce the following definitions:

\begin{definition} [Mean-value stabilization matrix]
\begin{equation}
\label{e:sigmaMV}
\bm{{\sigma}}_{MV} (\bm{v}_1 , \bm{v}_2 ; \bm{n}) := \int_0^1 (1 - \theta) \, \Big( \big| \tilde{\bm{A}}_n \big(\tilde{\bm{v}}(\theta ; \bm{v}_1 , \bm{v}_2) ; \bm{n} \big) \big|_{\tilde{\bm{A}}_0} + \big| \tilde{\bm{A}}_n \big(\tilde{\bm{v}}(\theta ; \bm{v}_2 , \bm{v}_1) ; \bm{n} \big) \big|_{\tilde{\bm{A}}_0} \Big) \, d \theta . 
\end{equation}
\end{definition}

\begin{definition}[Symmetric variable stabilization matrix]
\begin{equation}
\label{e:symStab}
\bm{\sigma}_{S}(\bm{v}_1 , \bm{v}_2 ; \bm{n}) := \big| \tilde{\bm{A}}_n (\bm{v}_*) \big|_{\tilde{\bm{A}}_0} , 
\end{equation}
where $\bm{v}_*$ is some state that depends on $\bm{v}_1$ and $\bm{v}_2$.
\end{definition}

\begin{definition}[Symmetric Lax-Friedrichs stabilization matrix]
\begin{equation}
\label{e:sLxF}
\bm{\sigma}_{SLF}(\bm{v}_1 , \bm{v}_2 ; \bm{n}) := \lambda_{max} (\bm{v}_*) \, \tilde{\bm{A}}_0 (\bm{v}_*) , 
\end{equation}
where $\lambda_{max}$ denotes the maximum-magnitude eigenvalue of $\bm{A}_n$, and $\bm{v}_*$ is some state that depends on $\bm{v}_1$ and $\bm{v}_2$.
\end{definition}

\begin{definition}[Entropy-stable stabilization matrix]\label{d:entrStableStabMatrix} A stabilization matrix satisfying
\begin{equation}
\label{e:entrStStabMatrix}
(\bm{v}_2 - \bm{v}_1)^t \cdot \bm{{\sigma}} (\bm{v}_1 , \bm{v}_2 ; \bm{n}) \cdot (\bm{v}_2 - \bm{v}_1) \geq (\bm{v}_2 - \bm{v}_1)^t \cdot \bm{\Sigma} (\bm{v}_1 , \bm{v}_2 ; \bm{n}) \cdot (\bm{v}_2 - \bm{v}_1) , 
\end{equation}
for all $\bm{v}_1 , \bm{v}_2 \in X_v$ and all $\norm{\bm{n}} = 1$, is said to be entropy stable, where
\begin{equation}
\begin{split}
\label{e:Lambda_n}
\bm{\Sigma} (\bm{v}_1 , \bm{v}_2 ; \bm{n}) := & \int_0^1 (1 - \theta) \, \Big( \tilde{\bm{A}}_n \big(\tilde{\bm{v}}(\theta ; \bm{v}_1 , \bm{v}_2) ; \bm{n} \big) - \tilde{\bm{A}}_n \big(\tilde{\bm{v}}(\theta ; \bm{v}_2 , \bm{v}_1) ; \bm{n} \big) \Big) \, d \theta \\
= & \int_0^1 (1 - 2 \theta) \ \tilde{\bm{A}}_n \big(\tilde{\bm{v}}(\theta ; \bm{v}_1 , \bm{v}_2) ; \bm{n} \big) \, d \theta = - \int_0^1 (1 - 2 \theta) \ \tilde{\bm{A}}_n \big(\tilde{\bm{v}}(\theta ; \bm{v}_2 , \bm{v}_1) ; \bm{n} \big) \, d \theta . 
\end{split}
\end{equation}
\end{definition}

{\bf Remark 1:} The last two equalities in Eq. \eqref{e:Lambda_n} follow from the change of variable $\theta' = 1 - \theta$ applied to the second and first terms in the integrand, respectively.


{\bf Remark 2:} Examples of entropy-stable stabilization matrices include the mean-value stabilization matrix \eqref{e:sigmaMV}, as well as the symmetric variable \eqref{e:symStab} and symmetric Lax-Friedrichs \eqref{e:sLxF} matrices with $\bm{v}_*$ such that
\begin{equation}
\label{e:cond_v*}
(\bm{v}_2 - \bm{v}_1)^t \cdot \big| \tilde{\bm{A}}_n (\bm{v}_*) \big|_{\tilde{\bm{A}}_0} \cdot (\bm{v}_2 - \bm{v}_1) \geq \sup_{\theta \in [0,1]} (\bm{v}_2 - \bm{v}_1)^t \cdot \big| \tilde{\bm{A}}_n \big( \tilde{\bm{v}} (\theta ; \bm{v}_1 , \bm{v}_2) \big) \big|_{\tilde{\bm{A}}_0} \cdot (\bm{v}_2 - \bm{v}_1) . 
\end{equation}


\begin{definition}[Projection viscous numerical flux]\label{d:projectionViscNumFlux} 
\begin{equation}
\label{e:numericalFluxNS_proj}
\widehat{\bm{g}}^{\Pi}_h (\bm{x}) := \Big( \Pi_{\bm{\mathcal{Q}}_h^k} \big[ \bm{G}(\bm{v}_h , \bm{q}_h) \big] \Big) (\bm{x}) \cdot \bm{n} (\bm{x}) , 
\end{equation}
where $\bm{x} \in F$, $F \in \partial \mathcal{T}_h$, and $\Pi_{\bm{\mathcal{Q}}_h^k}$ is the projection operator onto $\bm{\mathcal{Q}}_h^k$.
\end{definition}

{\bf Remark 3:} $\widehat{\bm{g}}^{\Pi}_h$ is not a local operator in the sense that $\widehat{\bm{g}}^{\Pi}_h (\bm{x})$ does not depend only on $(\bm{q}_h(\bm{x}),\bm{v}_h(\bm{x}),\bm{\widehat{v}}_h(\bm{x}))$ but on the solution over the entire element.

\begin{definition}[Entropy-stable viscous numerical flux]\label{d:entrStableViscNumFlux} A viscous numerical flux satisfying
\begin{equation}
\label{e:entrStableViscNumFlux}
\big( \bm{v}_h(\bm{x}) - \bm{\widehat{v}}_h(\bm{x}) \big)^t \cdot \widehat{\bm{g}}_h (\bm{x}) \geq \big( \bm{v}_h(\bm{x}) - \bm{\widehat{v}}_h(\bm{x}) \big)^t \cdot \Big[ \Big( \Pi_{\bm{\mathcal{Q}}_h^k} \big[ \bm{G}(\bm{v}_h,\bm{q}_h) \big] \Big) (\bm{x}) \cdot \bm{n} (\bm{x}) \Big] , 
\end{equation}
for all $(\bm{q}_h,\bm{v}_h,\bm{\widehat{v}}_h) \in \bm{\mathcal{Q}}_h^k \otimes \bm{\mathcal{V}}_h^k \otimes \bm{\mathcal{M}}_h^k$ and all $\bm{x} \in F \ \forall F \in \partial \mathcal{T}_h$, is said to be entropy stable.
\end{definition}

\subsection{\label{s:timeEvEuler}Time evolution of total entropy for the compressible Euler equations}

In this section, as well as in Sections \ref{s:entStEuler}, \ref{s:timeEvNS} and \ref{s:entStNS}, the numerical solution at time $t$ is assumed to be physical in the sense that the density and pressure are pointwise positive.

\begin{prop}\label{semiDiscGlobEntrIdentity} The time evolution of total generalized entropy in the entropy-variable hybridized DG discretization of the compressible Euler equations \eqref{IEDG} is given by
\begin{equation}
\label{entrEqhDG}
\begin{split}
\frac{d}{dt} \int_{\mathcal{T}_h} H (\bm{u}(\bm{v}_h)) & + \frac{1}{2} \int_{\partial \mathcal{T}_h} \big( \bm{v}_h - \widehat{\bm{v}}_h \big)^t \cdot \Big[ \bm{{\sigma}}(\widehat{\bm{v}}_h , \bm{v}_h ; \bm{n}) - \bm{\Sigma} (\widehat{\bm{v}}_h , \bm{v}_h ; \bm{n}) \Big] \cdot \big( \bm{v}_h - \widehat{\bm{v}}_h \big) \\
& + \mathcal{B}_{\partial \Omega}(\widehat{\bm{v}}_h , \bm{v}_h ; \bm{v}^{\partial \Omega}) = 0 , 
\end{split}
\end{equation}
where $\bm{\Sigma}(\widehat{\bm{v}}_h , \bm{v}_h ; \bm{n})$ is defined in Eq. \eqref{e:Lambda_n}, and
\begin{equation}
\label{e:Binv}
\mathcal{B}_{\partial \Omega}(\widehat{\bm{v}}_h , \bm{v}_h ; \bm{v}^{\partial \Omega}) = \int_{\partial \Omega} \bm{\mathcal{F}}_n (\widehat{\bm{v}}_h) - \int_{\partial \Omega} \widehat{\bm{v}}_h^t \cdot \bm{F}_n(\widehat{\bm{v}}_h) + \int_{\partial \Omega} \widehat{\bm{v}}_h^t \cdot \big( \widehat{\bm{f}}_h - \widehat{\bm{b}}_h (\widehat{\bm{v}}_h,\bm{v}_h ; \bm{v}^{\partial \Omega}) \big) 
\end{equation}
is a boundary term.
\end{prop}
\begin{proof} The proof is given in \ref{s:appProof}.
\end{proof}

\begin{corol}\label{semiDiscGlEntIdMeanValue} The time evolution of total generalized entropy for the compressible Euler equations with the mean-value stabilization matrix \eqref{e:sigmaMV} is given by
\begin{equation}
\label{entrEqhDG_MV}
\begin{split}
\frac{d}{dt} & \int_{\mathcal{T}_h} H (\bm{u}(\bm{v}_h)) \\
+ & \int_{\partial \mathcal{T}_h} \int_0^1 (1 - \theta) \, \big( \bm{v}_h - \widehat{\bm{v}}_h \big)^t \cdot \Big( \bm{A}_{n,\tilde{A}_0}^{+} \big( \tilde{\bm{v}}( \theta; \bm{v}_h , \widehat{\bm{v}}_h ) ; \bm{n} \big) - \bm{A}_{n,\tilde{A}_0}^{-} \big( \tilde{\bm{v}}( \theta; \widehat{\bm{v}}_h , \bm{v}_h ) ; \bm{n} \big) \Big) \cdot \big( \bm{v}_h - \widehat{\bm{v}}_h \big) \, d \theta \\
+ & \ \mathcal{B}_{\partial \Omega}(\widehat{\bm{v}}_h , \bm{v}_h ; \bm{v}^{\partial \Omega}) = 0 , 
\end{split}
\end{equation}
where
\begin{subequations}
\label{e:AsCorol}
\begin{equation}
\bm{A}_{n,\tilde{A}_0}^{+} ( \bm{v} ; \bm{n}) = \bm{A}_n^+ ( \bm{v} ; \bm{n}) \, \tilde{\bm{A}}_0 ( \bm{v} ; \bm{n}) = \frac{1}{2} \Big( \bm{A}_n ( \bm{v} ; \bm{n}) + \big| \bm{A}_n ( \bm{v} ; \bm{n}) \big| \Big) \tilde{\bm{A}}_0 ( \bm{v} ; \bm{n}) = \frac{1}{2} \Big( \tilde{\bm{A}}_n ( \bm{v} ; \bm{n}) + \big| \tilde{\bm{A}}_n ( \bm{v} ; \bm{n}) \big|_{\tilde{\bm{A}}_0} \Big) , 
\end{equation}
\begin{equation}
\bm{A}_{n,\tilde{A}_0}^{-} ( \bm{v} ; \bm{n}) = \bm{A}_n^- ( \bm{v} ; \bm{n}) \, \tilde{\bm{A}}_0 ( \bm{v} ; \bm{n}) = \frac{1}{2} \Big( \bm{A}_n ( \bm{v} ; \bm{n}) - \big| \bm{A}_n ( \bm{v} ; \bm{n}) \big| \Big) \tilde{\bm{A}}_0 ( \bm{v} ; \bm{n}) = \frac{1}{2} \Big( \tilde{\bm{A}}_n ( \bm{v} ; \bm{n}) - \big| \tilde{\bm{A}}_n ( \bm{v} ; \bm{n}) \big|_{\tilde{\bm{A}}_0} \Big) . 
\end{equation}
\end{subequations}
\end{corol}
\begin{proof} The desired result follows by combining Equations \eqref{e:sigmaMV}, \eqref{e:Lambda_n} and \eqref{entrEqhDG}, and noting that
\begin{equation}
\label{e:noting}
\begin{split}
\frac{1}{2} \int_0^1 & \Big( \bm{{\sigma}}_{MV} (\widehat{\bm{v}}_h , \bm{v}_h ; \bm{n}) - \bm{\Sigma} (\widehat{\bm{v}}_h , \bm{v}_h ; \bm{n}) \Big) \, d \theta \\
= \ & \frac{1}{2} \int_0^1 \Big( \big| \tilde{\bm{A}}_n \big( \tilde{\bm{v}} ( \theta ; \bm{v}_h , \widehat{\bm{v}}_h ) ; \bm{n} \big) \big|_{\tilde{\bm{A}}_0} + \big(1 - 2 \theta \big) \tilde{\bm{A}}_n \big( \tilde{\bm{v}} ( \theta ; \bm{v}_h , \widehat{\bm{v}}_h ) ; \bm{n} \big) \Big) \, d \theta \\
= \ & \frac{1}{2} \int_0^1 \big(1 - \theta) \, \Big( \tilde{\bm{A}}_n \big( \tilde{\bm{v}} ( \theta ; \bm{v}_h , \widehat{\bm{v}}_h ) ; \bm{n} \big) + \big| \tilde{\bm{A}}_n \big( \tilde{\bm{v}} ( \theta ; \bm{v}_h , \widehat{\bm{v}}_h ) ; \bm{n} \big) \big|_{\tilde{\bm{A}}_0} \Big) \, d \theta \\
- \ & \frac{1}{2} \int_0^1 \theta \, \Big( \tilde{\bm{A}}_n \big( \tilde{\bm{v}} ( \theta ; \bm{v}_h , \widehat{\bm{v}}_h ) ; \bm{n} \big) - \big| \tilde{\bm{A}}_n \big( \tilde{\bm{v}} ( \theta ; \bm{v}_h , \widehat{\bm{v}}_h ) ; \bm{n} \big) \big|_{\tilde{\bm{A}}_0} \Big) \, d \theta \\
= \ & \frac{1}{2} \int_0^1 \big(1 - \theta) \, \Big( \tilde{\bm{A}}_n \big( \tilde{\bm{v}} ( \theta ; \bm{v}_h , \widehat{\bm{v}}_h ) ; \bm{n} \big) + \big| \tilde{\bm{A}}_n \big( \tilde{\bm{v}} ( \theta ; \bm{v}_h , \widehat{\bm{v}}_h ) ; \bm{n} \big) \big|_{\tilde{\bm{A}}_0} \Big) \, d \theta \\
- \ & \frac{1}{2} \int_0^1 \big( 1 - \theta \big) \, \Big( \tilde{\bm{A}}_n \big( \tilde{\bm{v}} ( \theta ; \widehat{\bm{v}}_h , \bm{v}_h ) ; \bm{n} \big) - \big| \tilde{\bm{A}}_n \big( \tilde{\bm{v}} ( \theta ; \widehat{\bm{v}}_h , \bm{v}_h ) ; \bm{n} \big) \big|_{\tilde{\bm{A}}_0} \Big) \, d \theta \\
= \ & \int_0^1 (1 - \theta) \, \Big( \bm{A}_{n,\tilde{A}_0}^{+} \big( \tilde{\bm{v}}( \theta; \bm{v}_h , \widehat{\bm{v}}_h ) ; \bm{n} \big) - \bm{A}_{n,\tilde{A}_0}^{-} \big( \tilde{\bm{v}}( \theta; \widehat{\bm{v}}_h , \bm{v}_h ) ; \bm{n} \big) \Big) \, d \theta . 
\end{split}
\end{equation}
\end{proof}

\subsection{\label{s:entStEuler}Entropy stability for the compressible Euler equations}

\begin{theorem}[Semi-discrete entropy stability for the compressible Euler equations]
\label{mainTh} The entropy-variable hybridized DG discretization of the compressible Euler equations \eqref{IEDG} on a periodic domain with stabilization matrix as in \eqref{e:entrStStabMatrix} is entropy stable, that is, the total generalized entropy is non-increasing in time,
\begin{equation}
\label{e:entropyStability}
\frac{d}{dt} \int_{\mathcal{T}_h} H ( \bm{u}(\bm{v}_h(t)) ) \leq 0 , 
\end{equation}
and the following entropy bounds are satisfied
\begin{equation}
\label{e:entropyBound}
\int_{\mathcal{T}_h} H ( \bm{u}^*(\bm{v}_{h,0}) ) \leq \int_{\mathcal{T}_h} H ( \bm{u}(\bm{v}_h(t)) ) \leq \int_{\mathcal{T}_h} H ( \bm{u}(\bm{v}_{h,0}) ) , 
\end{equation}
where $\bm{u}^*$ is called the minimum total entropy state and is defined as
\begin{equation}
\label{e:uStar}
\bm{u}^*(\bm{v}_{h}) := \frac{1}{\mu_L (\Omega)} \int_{\mathcal{T}_h} \bm{u}(\bm{v}_{h}) , 
\end{equation}
and $\mu_L ( \Omega )$ denotes the Lebesgue measure of $\Omega$.
\end{theorem}
\begin{proof} The proof is given in \ref{s:appProof}.
\end{proof}

\begin{corol} The entropy-variable hybridized DG discretization of the compressible Euler equations on a periodic domain with either (i) the mean-value stabilization matrix \eqref{e:sigmaMV}, (ii) the symmetric variable stabilization matrix \eqref{e:symStab} with $\bm{v}_*$ as in Eq. \eqref{e:cond_v*}, or (iii) the symmetric Lax-Friedrichs stabilization matrix \eqref{e:sLxF} with $\bm{v}_*$ as in Eq. \eqref{e:cond_v*}, is entropy stable.
\end{corol}


From a mathematical perspective, Theorem \ref{mainTh} implies the scheme is unconditionally entropy stable in the sense that entropy stability holds for any polynomial order $k \geq 0$ and any non-singular finite element mesh $\mathcal{T}_h$. 
From a physical perspective, Theorem \ref{mainTh} implies the numerical solution satisfies the integral version of the Second Law of Thermodynamics in the problem domain $\Omega$. 
Also, Theorem \ref{mainTh} provides sufficient, but not necessary, conditions for entropy stability. 
Finally, we note that periodic boundary conditions are an adequate choice of boundary conditions to characterize the entropy stability of the scheme, and we recall that the development of entropy-stable boundary conditions is beyond the scope of this paper.

\subsection{\label{s:timeEvNS}Time evolution of total entropy for the compressible Navier-Stokes equations}

\begin{prop}\label{semiDiscGlobEntrIdentityNS} The time evolution of total generalized entropy in the entropy-variable hybridized DG discretization of the compressible Navier-Stokes equations \eqref{e:hDG_NS} on a non-curved (i.e.\ $p=1$) mesh is given by
\begin{equation}
\label{entrEqhDG_NS}
\begin{split}
\frac{d}{dt} \int_{\mathcal{T}_h} H (\bm{u}(\bm{v}_h)) & + \int_{\mathcal{T}_h} \bm{q}_h^t \cdot \widetilde{\bm{\mathcal{K}}} (\bm{v}_h) \cdot \bm{q}_h \\
& + \frac{1}{2} \int_{\partial \mathcal{T}_h} \big( \bm{v}_h - \widehat{\bm{v}}_h \big)^t \cdot \Big[ \bm{{\sigma}}(\widehat{\bm{v}}_h , \bm{v}_h ; \bm{n}) - \bm{\Sigma} (\widehat{\bm{v}}_h , \bm{v}_h ; \bm{n}) \Big] \cdot \big( \bm{v}_h - \widehat{\bm{v}}_h \big) \\
& + \int_{\partial \mathcal{T}_h} \big( \bm{v}_h - \widehat{\bm{v}}_h \big)^t \cdot \Big[ \widehat{\bm{g}}_h - \Big( \Pi_{\bm{\mathcal{Q}}_h^k} \big[ \bm{G}(\bm{v}_h , \bm{q}_h) \big] \cdot \bm{n} \Big) \Big] \\
& + \mathcal{B}_{\partial \Omega}(\widehat{\bm{v}}_h , \bm{v}_h , \bm{q}_h ; \bm{v}^{\partial \Omega}) = 0 , 
\end{split}
\end{equation}
where 
$\bm{\Sigma}(\widehat{\bm{v}}_h , \bm{v}_h ; \bm{n})$ is defined in Eq. \eqref{e:Lambda_n}, $\widetilde{\bm{\mathcal{K}}}_{ij} (\bm{v}_h) = \bm{\mathcal{K}}_{ij} (\bm{u}(\bm{v}_h)) \, \tilde{\bm{A}}_0 (\bm{v}_h) , \, i,j=1,\dots,d$ are symmetric positive semi-definite, and
\begin{equation}
\label{e:Binv_NS}
\mathcal{B}_{\partial \Omega}(\widehat{\bm{v}}_h , \bm{v}_h , \bm{q}_h ; \bm{v}^{\partial \Omega}) = \int_{\partial \Omega} \bm{\mathcal{F}}_n (\widehat{\bm{v}}_h) - \int_{\partial \Omega} \widehat{\bm{v}}_h^t \cdot \bm{F}_n(\widehat{\bm{v}}_h) + \int_{\partial \Omega} \widehat{\bm{v}}_h^t \cdot \big( \widehat{\bm{f}}_h + \widehat{\bm{g}}_h - \widehat{\bm{b}}_h (\widehat{\bm{v}}_h,\bm{v}_h ; \bm{v}^{\partial \Omega}) \big) 
\end{equation}
is a boundary term.
\end{prop}
\begin{proof} The proof is given in \ref{s:appProof}.
\end{proof}

\begin{corol}\label{semiDiscGlEntIdMeanValue_NS} The time evolution of total generalized entropy for the compressible Navier-Stokes equations on a non-curved mesh with the mean-value stabilization matrix \eqref{e:sigmaMV} and the projection viscous numerical flux \eqref{e:numericalFluxNS_proj} is given by
\begin{equation}
\label{entrEqhDG_MV_NS}
\begin{split}
\frac{d}{dt} & \int_{\mathcal{T}_h} H (\bm{u}(\bm{v}_h)) + \int_{\mathcal{T}_h} \bm{q}_h^t \cdot \widetilde{\bm{\mathcal{K}}} (\bm{v}_h) \cdot \bm{q}_h \\
+ & \int_{\partial \mathcal{T}_h} \int_0^1 (1 - \theta) \, \big( \bm{v}_h - \widehat{\bm{v}}_h \big)^t \cdot \Big( \bm{A}_{n,\tilde{A}_0}^{+} \big( \tilde{\bm{v}}( \theta; \bm{v}_h , \widehat{\bm{v}}_h ) ; \bm{n} \big) - \bm{A}_{n,\tilde{A}_0}^{-} \big( \tilde{\bm{v}}( \theta; \widehat{\bm{v}}_h , \bm{v}_h ) ; \bm{n} \big) \Big) \cdot \big( \bm{v}_h - \widehat{\bm{v}}_h \big) \, d \theta \\
+ & \, \mathcal{B}_{\partial \Omega}(\widehat{\bm{v}}_h , \bm{v}_h , \bm{q}_h ; \bm{v}^{\partial \Omega}) = 0 , 
\end{split}
\end{equation}
where $\bm{A}_{n,\tilde{A}_0}^{+}$ and $\bm{A}_{n,\tilde{A}_0}^{-}$ are defined in Eq. \eqref{e:AsCorol}.
\end{corol}
\begin{proof} The desired result follows by combining Equations \eqref{e:sigmaMV}, \eqref{e:Lambda_n}, \eqref{e:numericalFluxNS_proj}, \eqref{e:noting} and \eqref{entrEqhDG_NS}.
\end{proof}

\subsection{\label{s:entStNS}Entropy stability for the compressible Navier-Stokes equations}

\begin{theorem}[Semi-discrete entropy stability for the compressible Navier-Stokes equations]
\label{mainThNS} The entropy-variable hybridized DG discretization of the compressible Navier-Stokes equations \eqref{e:hDG_NS} on a non-curved (i.e.\ $p=1$) mesh with periodic boundaries, stabilization matrix as in \eqref{e:entrStStabMatrix} and viscous numerical flux as in \eqref{e:entrStableViscNumFlux} is entropy stable, that is, the total generalized entropy is non-increasing in time,
\begin{equation}
\label{e:entropyStabilityNS}
\frac{d}{dt} \int_{\mathcal{T}_h} H ( \bm{u}(\bm{v}_h(t)) ) \leq 0 , 
\end{equation}
and the following entropy bounds are satisfied
\begin{equation}
\label{e:entropyBoundNS}
\int_{\mathcal{T}_h} H ( \bm{u}^*(\bm{v}_{h,0}) ) \leq \int_{\mathcal{T}_h} H ( \bm{u}(\bm{v}_h(t)) ) \leq \int_{\mathcal{T}_h} H ( \bm{u}(\bm{v}_{h,0}) ) , 
\end{equation}
where $\bm{u}^*$ is the minimum total entropy state defined in Eq. \eqref{e:uStar}.
\end{theorem}
\begin{proof} The proof is given in \ref{s:appProof}.
\end{proof}

\begin{corol} The entropy-variable hybridized DG discretization of the compressible Navier-Stokes equations on a non-curved mesh with periodic boundaries, the projection viscous numerical flux \eqref{e:numericalFluxNS_proj}, and with either (i) the mean-value stabilization matrix \eqref{e:sigmaMV}, (ii) the symmetric variable stabilization matrix \eqref{e:symStab} with $\bm{v}_*$ as in Eq. \eqref{e:cond_v*}, or (iii) the symmetric Lax-Friedrichs stabilization matrix \eqref{e:sLxF} with $\bm{v}_*$ as in Eq. \eqref{e:cond_v*}, is entropy stable.
\end{corol}

We note again that Theorem \ref{mainThNS} implies the numerical solution satisfies the integral version of the Second Law of Thermodynamics in $\Omega$, and that it provides sufficient, but not necessary, conditions for entropy stability. Also, entropy stability holds for any approximating polynomial order $k \geq 0$ and any non-singular, non-curved mesh. 
The need for the projection viscous numerical flux to ensure entropy stability suggests that if $\widehat{\bm{g}}_h$ is as in Eq. \eqref{e:numericalFluxNS_1}, it is when $\widetilde{\bm{\mathcal{K}}}(\bm{v}_h)$ significantly changes inside an element, and consequently $\bm{G}(\bm{v}_h , \bm{q}_h)$ is not well represented in $\bm{\mathcal{Q}}_h^k$, that the viscous terms may lead or contribute to numerical instability. A similar logic applies to other definitions of $\widehat{\bm{g}}_h$.


\subsection{$L^2$ stability}

We conclude this section with a quick note on the $L^2$ stability of the scheme.

\begin{prop}
If the scheme is entropy stable and $\tilde{\bm{A}}_0$ remains uniformly bounded in the sense that there exist positive constants $C \geq c > 0$ such that
\begin{equation}
\label{e:condL2}
c \, \norm{\bm{z}}_{L^2(\mathbb{R}^m)}^2 \leq \bm{z}^t \cdot \tilde{\bm{A}}_0 (\bm{v}) \cdot \bm{z} \leq C \, \norm{\bm{z}}_{L^2(\mathbb{R}^m)}^2 
\end{equation}
for all $\bm{z} \in \mathbb{R}^m$ and all $\bm{v} \in \mathcal{C}\mathcal{H}\big( \big\{ \bm{v}_h(x , t) , \ \bm{x} \in \mathcal{T}_h \ t \in [0,t_f] \big\}\big)$, where $\mathcal{C}\mathcal{H} ( \, \cdot \, )$ denotes convex hull, then the scheme is $L^2$ stable in the sense that
\begin{equation}
\label{e:L2stab}
\norm{ \bm{u}(\bm{v}_h(t)) - \bm{u}^*(\bm{v}_{h,0}) }_{[L^2(\Omega)]^m} \leq \bigg( \frac{C}{c} \bigg)^{1/2} \, \norm{ \bm{u}(\bm{v}_{h,0}) - \bm{u}^*(\bm{v}_{h,0}) }_{[L^2(\Omega)]^m} \quad \textnormal{for all} \quad 0 \leq t \leq t_f . 
\end{equation}
This result applies both to the compressible Euler and Navier-Stokes equations.
\end{prop}
\begin{proof} For any $0 \leq t \leq t_f$, it holds that
\begin{equation}
\label{e:L2stabProof}
\begin{split}
\frac{c}{2} & \norm{ \bm{u}(\bm{v}_h(t)) - \bm{u}^*(\bm{v}_{h,0}) }_{[L^2(\Omega)]^m}^2 \leq \frac{1}{2} \int_{\mathcal{T}_h} \big( \bm{u}(\bm{v}_h(t)) - \bm{u}^*(\bm{v}_{h,0}) \big)^t \cdot \frac{\partial^2 H (\bar{\bm{u}}(x,t))}{\partial \bm{u}^2} \cdot \big( \bm{u}(\bm{v}_h(t)) - \bm{u}^*(\bm{v}_{h,0}) \big) \\ 
& \qquad = \int_{\mathcal{T}_h} H(\bm{u}(\bm{v}_h(t))) - \int_{\mathcal{T}_h} H(\bm{u}^*(\bm{v}_{h,0})) \leq \int_{\mathcal{T}_h} H(\bm{u}(\bm{v}_{h,0})) - \int_{\mathcal{T}_h} H(\bm{u}^*(\bm{v}_{h,0})) \\
& \qquad = \frac{1}{2} \int_{\mathcal{T}_h} \big( \bm{u}(\bm{v}_{h,0}) - \bm{u}^*(\bm{v}_{h,0}) \big)^t \cdot \frac{\partial^2 H (\bar{\bm{u}}_0(x))}{\partial \bm{u}^2} \cdot \big( \bm{u}(\bm{v}_{h,0}) - \bm{u}^*(\bm{v}_{h,0}) \big) \\
& \qquad \leq \frac{C}{2} \norm{ \bm{u}(\bm{v}_{h,0}) - \bm{u}^*(\bm{v}_{h,0}) }_{[L^2(\Omega)]^m}^2 , 
\end{split}
\end{equation}
where $\bar{\bm{u}}( \cdot ,t)$ and $\bar{\bm{u}}_0( \cdot )$ are in the convex hull of $\bm{u}(\bm{v}_h( \cdot , t))$ and $\bm{u}(\bm{v}_{h,0}( \cdot ))$, respectively. The equalities in \eqref{e:L2stabProof} follow from the Taylor series with remainder in Lagrange form\footnote{Note that Eq. \eqref{e:condL2} implies $H$ is sufficiently regular for this form of the Remainder Theorem to apply everywhere in space and time. Also, the Remainder Theorem ensures $\bar{\bm{u}}( \cdot ,t)$ and $\bar{\bm{u}}_0( \cdot )$ are in the convex hull of $\bm{u}(\bm{v}_h( \cdot , t))$ and $\bm{u}(\bm{v}_{h,0}( \cdot ))$.} and the definition of $\bm{u}^*$; the second inequality follows from entropy stability; and the first and third inequalities from \eqref{e:condL2}. Equation \eqref{e:L2stab} then trivially follows.
\end{proof}
%
%
%
%

\section{\label{s:numExamples}Numerical examples}

We present a series of numerical examples to illustrate the convergence rate and stability of the proposed family of schemes. Steady and unsteady flows in subsonic, transonic, and supersonic regimes are considered. For some of the test problems, the performance and robustness is compared to the conservation-variable hybridized DG methods \cite{Fernandez:17a,Nguyen:15}. In all the examples, the Lax-Friedrichs stabilization matrix in \cite{Fernandez:AIAA:17a} is used for the conservation-variable schemes, and the symmetric Lax-Friedrichs stabilization matrix in Eq. \eqref{e:sLxF} with $\bm{v}_* = \widehat{\bm{v}}_h$ is used for the entropy-variable schemes. Equation \eqref{e:numericalFluxNS_2} is used for the viscous numerical flux. Characteristics-based, non-reflecting boundary conditions are prescribed on the inflow/outflow boundaries, slip wall boundary conditions are used on the solid surfaces for the inviscid problems, and no-slip, adiabatic boundary conditions are used on the solid surfaces for the viscous problems. For the unsteady examples, the semi-discrete system \eqref{IEDG} is integrated in time using the third-order, three-stage $L$-stable diagonally implicit Runge-Kutta DIRK(3,3) method \cite{Alexander:77}. For the steady-state examples, the backward Euler method is used for pseudo-time marching towards the steady state. Finally, $Pr = 0.71$, $\gamma = 1.4$ and $\beta = 0$ are assumed in all the test problems.

We note that entropy stability is not necessarily preserved upon DIRK(3,3) time discretization. Also, the symmetric Lax-Friedrichs stabilization matrix with $\bm{v}_* = \widehat{\bm{v}}_h$ may not satisfy Eq. \eqref{e:entrStStabMatrix} for pathological choices of the numerical solution $(\bm{v}_h , \widehat{\bm{v}}_h)$, and similarly the viscous numerical flux \eqref{e:numericalFluxNS_2} may not satisfy \eqref{e:entrStableViscNumFlux}. This stabilization matrix and viscous numerical flux are considered since they lead to more computationally efficient implementations. In addition, the numerical scheme remains entropy stable in practice, as illustrated in the numerical examples in this section.



\subsection{Ringleb flow}

This example is aimed at verifying the optimal accuracy order of the entropy-variable hybridized DG methods for smooth solutions of the Euler equations. 
To that end, we consider the Ringleb flow, an exact smooth solution of the two-dimensional steady-state Euler equations obtained using the hodograph method \cite{Chiocchia:1985}. The radial velocity $V$ at $(x,y)$ is given by the following nonlinear equation
\begin{equation}
\label{e:P1eq1}
(x-0.5A^2) + y^2 = \frac{1}{4 c^{10} V^4} , 
\end{equation}
where
\begin{equation}
c = \sqrt{1-\frac{V^2}{5}} , \qquad \qquad A = \frac{1}{c} + \frac{1}{3 c^2} + \frac{1}{5 c^5} - \frac{1}{2} \ln \bigg( \frac{1+c}{1-c} \bigg) . 
\end{equation}
The exact solution is then computed as
\begin{equation}
\label{e:RinglebExactSol}
\rho = c^5 , \qquad \qquad \rho V_x = c^5 V \cos \theta , \qquad \qquad \rho V_y = c^5 V \sin \theta , \qquad \qquad \rho E = \frac{c^7}{\gamma (\gamma-1)} + \frac{1}{2} c^5 V^2 , 
\end{equation}
where
\begin{equation}
\label{e:P1eq4}
\theta = \arcsin (\psi V) , \qquad \qquad \psi = \sqrt{\frac{1}{2 V^2} - c^5 \, (x -0.5 A)} . 
\end{equation}
All quantities in Equations \eqref{e:P1eq1}$-$\eqref{e:P1eq4} are given in non-dimensional form. Note the rightmost equation in \eqref{e:RinglebExactSol} assumes the reference density $\rho_0$, velocity $V_0$ and total energy $\rho E_0$ for non-dimensionalization are related through $\rho E_0 = \rho_0 V_0^2$. For other choices of non-dimensionalization, the expression for $\rho E$ needs to be adapted accordingly.

In particular, we consider the Ringleb flow on the domain $\Omega = (-5,-1) \otimes (1,5)$ and prescribe the exact solution as boundary condition on $\partial \Omega$. 
The computational domain is partitioned into triangular meshes that are obtained by splitting a uniform $n \times n$ Cartesian grid into $2 n^2$ triangles. 
We perform grid convergence studies for polynomial orders $k = 1$, $2$, $3$, $4$ using mesh resolutions $1 / h = 2$, $4$, $8$, $16$, where $h = 4 / n$ is the characteristic element size. Tables \ref{t:RinglebConvergence_HDG}, \ref{t:RinglebConvergence_IEDG} and \ref{t:RinglebConvergence_EDG} present the errors and convergence rates, measured in the $[L^2(\Omega)]^m$ norm of non-dimensional conservation variables, for the entropy-variable HDG, IEDG and EDG schemes, respectively. The numerical solution converges to the exact solution with accuracy order of $k + 1$ for the three schemes; which is the optimal accuracy order given by the polynomial approximation space.

\begin{table}
\centering
\begin{tabular}{|c||cc|cc|cc|cc|}
\hline
 Mesh & \multicolumn{2}{c|}{$k = 1$} & \multicolumn{2}{c|}{$k = 2$} & \multicolumn{2}{c|}{$k = 3$} & \multicolumn{2}{c|}{$k = 4$} \\
 $1/h$ & Error & Order & Error & Order & Error & Order & Error & Order \\
\hline
2 & 4.92e-3 & $-$ & 6.79e-4 & $-$ & 3.94e-5 & $-$ & 4.79e-6 & $-$ \\
4 & 1.27e-3 & 1.95 & 1.05e-4 & 2.69 & 3.29e-6 & 3.58 & 2.01e-7 & 4.57 \\
8 & 3.26e-4 & 1.96 & 1.61e-5 & 2.71 & 1.83e-7 & 4.17 & 7.46e-9 & 4.75 \\
16 & 8.25e-5 & 1.98 & 2.21e-6 & 2.86 & 1.26e-8 & 3.86 & 2.67e-10 & 4.80\\
\hline
\end{tabular}
\caption{\label{t:RinglebConvergence_HDG} History of convergence of the entropy-variable HDG method for the Ringleb flow.}
\end{table}


\begin{table}
\centering
\begin{tabular}{|c||cc|cc|cc|cc|}
\hline
 Mesh & \multicolumn{2}{c|}{$k = 1$} & \multicolumn{2}{c|}{$k = 2$} & \multicolumn{2}{c|}{$k = 3$} & \multicolumn{2}{c|}{$k = 4$} \\
 $1/h$ & Error & Order & Error & Order & Error & Order & Error & Order \\
\hline
2 & 5.49e-3 & $-$ & 9.01e-4 & $-$ & 4.20e-5 & $-$ & 5.65e-6 & $-$ \\
4 & 1.42e-3 & 1.95 & 1.59e-4 & 2.50 & 3.38e-6 & 3.64 & 2.56e-7 & 4.47 \\
8 & 3.62e-4 & 1.97 & 2.56e-5 & 2.64 & 1.74e-7 & 4.28 & 9.90e-9 & 4.70 \\
16 & 9.30e-5 & 1.96 & 3.50e-6 & 2.87 & 1.18e-8 & 3.88 & 3.70e-10 & 4.74 \\
\hline
\end{tabular}
\caption{\label{t:RinglebConvergence_IEDG} History of convergence of the entropy-variable IEDG method for the Ringleb flow.}
\end{table}

\begin{table}
\centering
\begin{tabular}{|c||cc|cc|cc|cc|}
\hline
 Mesh & \multicolumn{2}{c|}{$k = 1$} & \multicolumn{2}{c|}{$k = 2$} & \multicolumn{2}{c|}{$k = 3$} & \multicolumn{2}{c|}{$k = 4$} \\
 $1/h$ & Error & Order & Error & Order & Error & Order & Error & Order \\
\hline
2 & 5.88e-3 & $-$ & 8.96e-4 & $-$ & 4.03e-5 & $-$ & 5.54e-6 & $-$ \\
4 & 1.46e-3 & 2.01 & 1.55e-4 & 2.53 & 3.25e-6 & 3.63 & 2.47e-7 & 4.49 \\
8 & 3.67e-4 & 2.00 & 2.47e-5 & 2.65 & 1.66e-7 & 4.30 & 9.46e-9 & 4.70 \\
16 & 9.21e-5 & 1.99 & 3.36e-6 & 2.88 & 1.15e-8 & 3.85 & 3.55e-10 & 4.74 \\
\hline
\end{tabular}
\caption{\label{t:RinglebConvergence_EDG} History of convergence of the entropy-variable EDG method for the Ringleb flow.}
\end{table}

\subsection{\label{s:couette}Couette flow}

The goal of this test problem is to verify the convergence rates for the Navier-Stokes equations. We consider a two-dimensional steady-state compressible Couette flow with a source term on a square domain $\Omega = (0,L)^2$. The exact solution is given by
\begin{equation}
\begin{split}
V_x & = V_0 \, \bar{y} \, \log(1+\bar{y}) , \\
V_y & = 0 , \\
p & = P_0 , \\
T & = T_0 \, \big( \alpha + \bar{y} \, (\beta - \alpha) \big) + \frac{V_0^2 \, Pr}{2 \, c_v \gamma} \, \bar{y} \, (1-\bar{y}) , 
\end{split}
\end{equation}
where $V_0$, $P_0$, $\alpha$ and $\beta$ are positive constants, $T_0 = V_0^2 / \big( (\gamma-1) \, \gamma \, M_0^2 \, c_v \big)$ is the temperature such that $V_0$ corresponds to a Mach number $M_0$, and $\bar{y} = y / L$ is a non-dimensional distance along the wall-normal direction. The density is given by the ideal gas law $\rho = p / \big( (\gamma-1) \, c_v \, T \big)$. The source term, appearing in the right-hand side of the Navier-Stokes equations \eqref{e:ns1}, is determined from the exact solution and reads as
\begin{equation}
\bm{s} = \Bigg( 0 , -\frac{\mu V_0}{L^2} \frac{2+\bar{y}}{(1+\bar{y})^2} , 0 , - \frac{\mu V_0^2}{L^2} \bigg( \log^2(1+\bar{y}) + \frac{\bar{y} \, \log(1+\bar{y})}{1+\bar{y}} + \frac{\bar{y} \, (3+2\bar{y}) \log(1+\bar{y}) - 2\bar{y} - 1}{(1+\bar{y})^2} \bigg) \Bigg)^t . 
\end{equation}
In particular, we take $\alpha = 0.8$, $\beta = 0.85$ and $M_0 = 0.15$, with $Pr = 0.71$ and $\gamma = 1.4$ as discussed before. While the exact solution is independent of the Reynolds number $Re_0 = \rho_0 V_0 L / \mu$, where $\rho_0 = \gamma \, P_0 \, M_0^2 / V_0^2$, we set it to $0.1$ in the simulations. This completes the non-dimensional description of the problem.

The computational domain is discretized into a uniform $n \times n$ Cartesian mesh and each square is further divided into two triangles. The exact solution is prescribed as boundary condition on $\partial \Omega$. 
In order to assess the accuracy of the numerical solution, we compute the $L^2$ norm of the error in the non-dimensional density $\rho^* = \rho / \rho_0$, momentum $\bm{p}^* = (\rho V_x , \rho V_y) / \rho_0 V_0$, total energy $\rho E^* = \rho E \, \rho_0^{-1} \, V_0^{-2}$, viscous stress tensor $\bm{\tau}^* = L \, \bm{\tau} \, \mu^{-1} V_0^{-1}$, and heat flux $\bm{f}^* = L \, \bm{f} \, \mu^{-1} V_0^{-1}$, for different mesh sizes and polynomial orders. 
Note that the approximate stresses and heat fluxes are computed from the approximate solution $\bm{v}_h$ and the approximate gradients $\bm{q}_h$.

Table \ref{t:CouetteConvergence_HDG} shows the errors and convergence rates for the entropy-variable HDG scheme. All the quantities converge with the optimal accuracy order of $k + 1$. The fact that HDG yields optimal accuracy both for the approximate solution and the approximate gradients is a very important advantage since most DG methods provide suboptimal convergence of order $k$ for the approximate gradients. All other schemes within the entropy-variable hybridized DG family, including IEDG and EDG, converge at the more common rates of $k+1$ for the approximate solution and $k$ for the approximate gradients (not shown).

\begin{table}
\centering
    \addtolength{\leftskip} {-2cm}
    \addtolength{\rightskip}{-2cm}
\begin{tabular}{|c|c||cc|cc|cc|cc|cc|}
\hline
Degree &  Mesh & \multicolumn{2}{c|}{$\norm{\rho_h^*-\rho^*}_{L^2(\Omega)}$} & \multicolumn{2}{c|}{$\norm{\bm{p}_h^*-\bm{p}^*}_{[L^2(\Omega)]^d}$} & \multicolumn{2}{c|}{$\norm{\rho E_h^*-\rho E^*}_{L^2(\Omega)}$} & \multicolumn{2}{c|}{$\norm{\bm{\tau}_h^*-\bm{\tau}^*}_{[L^2(\Omega)]^{d \times d}}$} & \multicolumn{2}{c|}{$\norm{\bm{f}_h^*-\bm{f}^*}_{[L^2(\Omega)]^d}$} \\
$k$ & $L/h$ & Error & Order & Error & Order & Error & Order & Error & Order & Error & Order \\
\hline
\multirow{4}{*}{1} & 8 & 8.28e-5 & $-$ & 7.79e-4 & $-$ & 6.00e-3 & $-$ & 3.28e-3 & $-$ & 6.65e-4 & $-$ \\
 & 16 & 1.69e-5 & 2.29  & 1.98e-4 & 1.98  & 1.32e-3 & 2.19  & 9.90e-4 & 1.73  & 1.68e-4 & 1.99 \\
 & 31 & 3.79e-6 & 2.16  & 5.01e-5 & 1.98  & 3.12e-4 & 2.08  & 2.95e-4 & 1.75  & 4.21e-5 & 1.99  \\
 & 64 & 1.06e-6 & 1.84  & 1.26e-5 & 1.99  & 8.50e-5 & 1.88  & 8.70e-5 & 1.76  & 1.06e-5 & 2.00  \\
\hline
\multirow{4}{*}{2} & 8 & 8.33e-6 & $-$ & 1.20e-5 & $-$ & 5.37e-4 & $-$ & 7.61e-5 & $-$ & 4.16e-6 & $-$ \\
 & 16 & 1.51e-6 & 2.47  & 1.52e-6 & 2.98  & 9.68e-5 & 2.47  & 1.15e-5 & 2.72  & 5.30e-7 & 2.97  \\
 & 32 & 2.49e-7 & 2.60  & 1.92e-7 & 2.98  & 1.59e-5 & 2.60  & 1.66e-6 & 2.80  & 6.72e-8 & 2.98  \\
 & 64 & 4.14e-8 & 2.59  & 2.43e-8 & 2.98  & 2.64e-6 & 2.59  & 2.37e-7 & 2.81  & 8.49e-9 & 2.99  \\
\hline
\multirow{4}{*}{3} & 8 & 1.18e-7 & $-$ & 2.30e-7 & $-$ & 7.63e-6 & $-$ & 2.28e-6 & $-$ & 5.73e-8 & $-$ \\
 & 16 & 8.31e-9 & 3.83  & 1.44e-8 & 4.00  & 5.37e-7 & 3.83  & 1.90e-7 & 3.58  & 3.77e-9 & 3.93 \\
 & 32 & 5.44e-10 & 3.93  & 9.70e-10 & 3.89  & 3.52e-8 & 3.93  & 1.73e-8 & 3.46  & 3.40e-10 & $\times$  \\
 & 64 & 3.48e-11 & 3.97  & 5.42e-11 & 4.16  & 2.25e-9 & 3.97  & 9.99e-10 & 4.12  & 4.26e-10 & $\times$ \\
\hline
\multirow{4}{*}{4} & 8 & 6.05e-9 & $-$ & 5.16e-9 & $-$ & 3.88e-7 & $-$ & 5.20e-8 & $-$ & 9.91e-10 & $-$ \\
 & 16 & 2.43e-10 & 4.64  & 2.21e-10 & 4.54  & 1.35e-8 & 4.84  & 2.21e-9 & 4.56 & 1.90e-10 & $\times$ \\
 & 32 & 7.55e-12 & 5.01  & 6.60e-12 & 5.07  & 3.47e-10 & 5.28  & 8.30e-11 & 4.73  & 2.59e-10 & $\times$ \\
 & 64 & 2.31e-13 & 5.03  & 2.04e-13 & 5.02  & 8.82e-12 & 5.30  & 8.04e-11 & $\times$  & 1.20e-10 & $\times$  \\
\hline
\end{tabular}
\caption{\label{t:CouetteConvergence_HDG} History of convergence of the entropy-variable HDG method for the Couette flow. The subscript $h$ denotes the numerical solution and no subscript denotes the exact solution. $\times$ indicates the convergence rate cannot be accurately computed due to finite precision arithmetic issues. Note finite precision affects before the errors in $\bm{\tau}^*$ and $\bm{f}^*$ than the errors in $\rho^*$, $\bm{p}^*$ and $\rho E^*$ due to the larger condition number to evaluate the former.}
\end{table}

\subsection{Inviscid supersonic flow past a circular duct}



This test case involves the steady supersonic flow in a two-dimensional channel with a 4\% thick circular bump on the bottom side. The length-to-height ratio of the channel is 3:1, the flow is inviscid, goes from left to right, and the inlet Mach number is $M_{\infty} = 1.4$. 
We use a finite element mesh with 2,400 isoparametric triangular elements and polynomials of degree $4$ to represent both the solution and the geometry.

Numerical results for entropy-variable IEDG without a shock capturing method are presented on the top images in Figure \ref{f:ductFigures}. No visual differences are observed in the numerical solution with the entropy-variable HDG and EDG schemes (not shown). Despite the severe Gibbs oscillations near the shock wave, the entropy-variable hybridized DG schemes are stable in all cases. Conservation-variable HDG, IEDG and EDG failed to converge for this problem without shock capturing. 
The solution with the entropy-variable IEDG scheme and the shock capturing method in \cite{Fernandez:SciTech:18a,Fernandez:JCPshock:18b} is shown at the bottom of Figure \ref{f:ductFigures}. The shock is well resolved and non-oscillatory in this case. We note that entropy stability can be shown to be preserved with this shock capturing method \cite{Fernandez:PhD:2018}.

 \begin{figure}
 \centering
 {\includegraphics[width=0.48\textwidth]{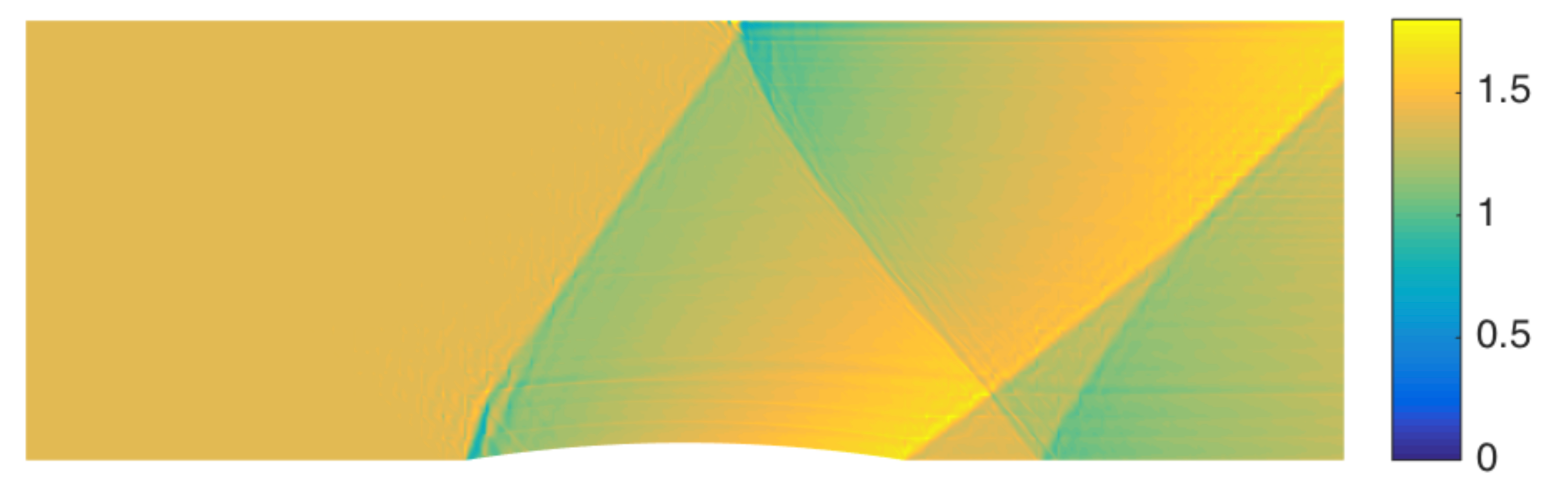}}
 \hfill {\includegraphics[width=0.48\textwidth]{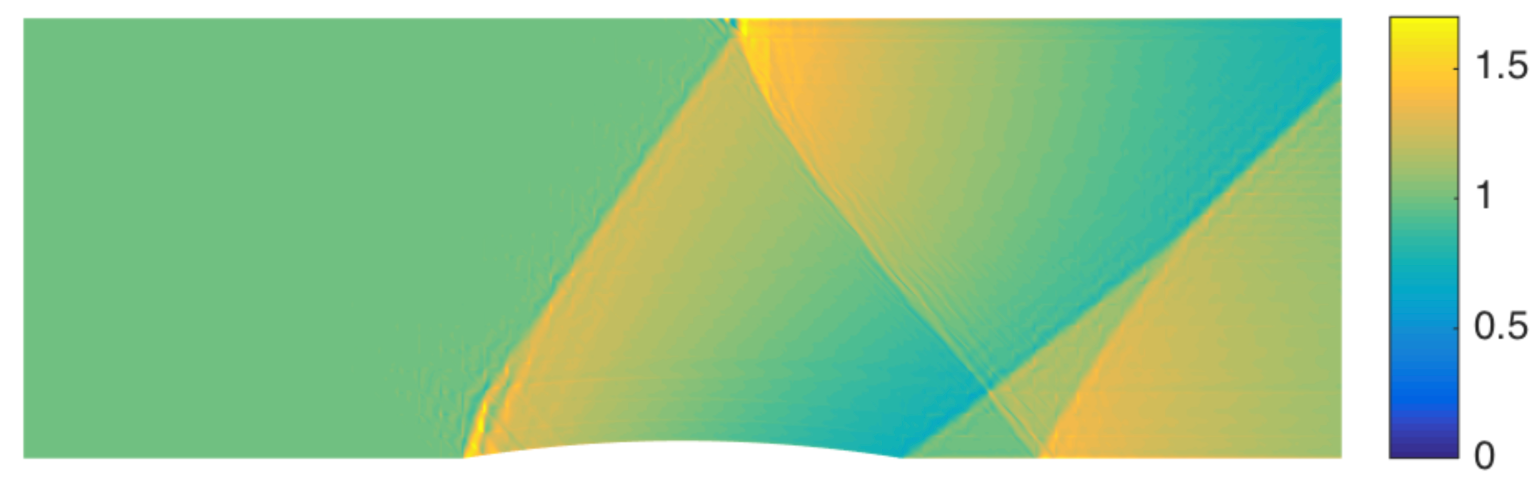}}
 \hfill {\includegraphics[width=0.48\textwidth]{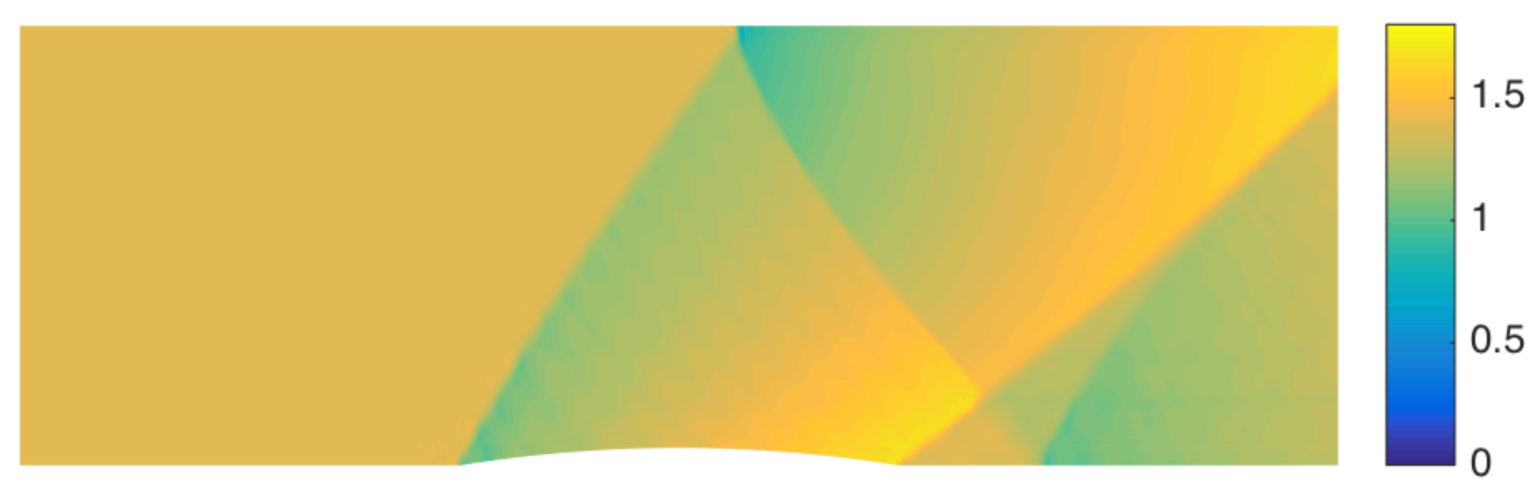}}
 \hfill {\includegraphics[width=0.48\textwidth]{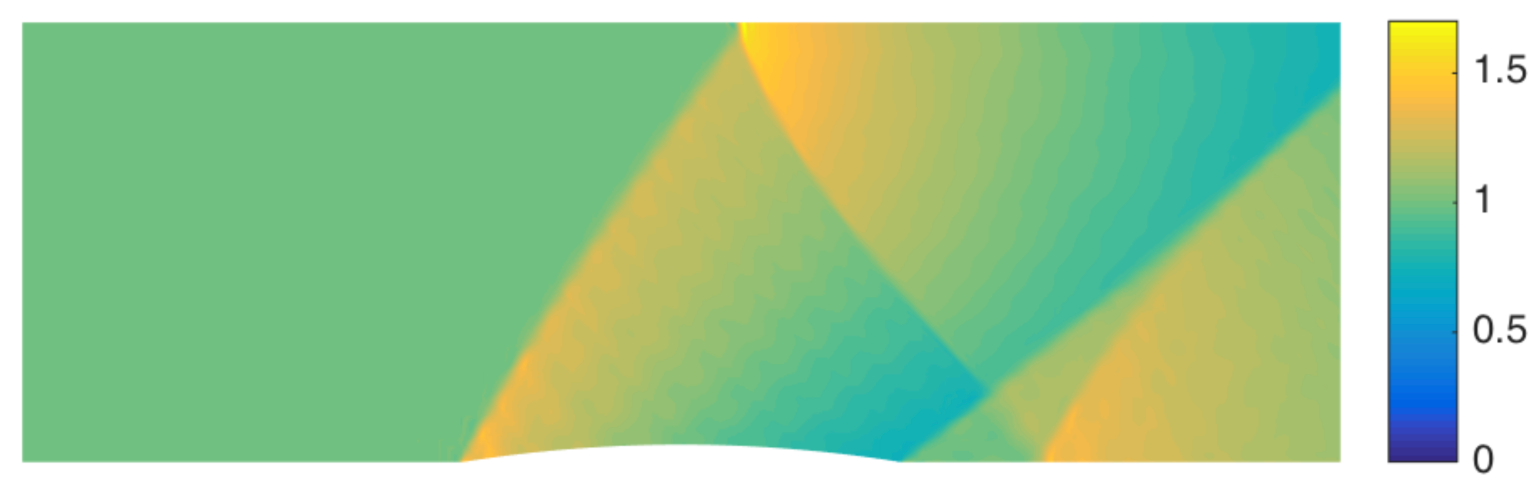}}
  \caption{Steady inviscid flow past a circular duct at $M_{\infty} = 1.4$. Figures show the Mach number (left) and non-dimensional density $\rho / \rho_{\infty}$ (right) for the entropy-variable IEDG method without (top) and with (bottom) shock capturing. No visual differences are observed between the entropy-variable HDG, IEDG and EDG solutions. Conservation-variable HDG, IEDG and EDG failed to converge for this problem without shock capturing.}\label{f:ductFigures}
 \end{figure}

%
%
%
%

\subsection{Shu vortex}



\subsubsection{Case description and numerical discretization}

The goal of this test problem is to investigate the stability of the scheme for the long time integration of vortex transport phenomena. To this end, we consider the two-dimensional inviscid vortex problem proposed by Shu\footnote{While this and other similar types of exact solutions of the Euler and Navier-Stokes equations have been known for a number of years (see, e.g., \cite{Colonius:1991}), we refer to this problem as the Shu vortex since, to our best knowledge, it was Shu \cite{Shu:98} that first proposed it to assess the advantages of high-order methods for long time simulations, and then other authors followed \cite{Castonguay:2011,Hesthaven:2008,Spiegel:2015,Vermeire:2013,Vincent:2011,Wang:2007,Wang:2013}.} \cite{Shu:98}. The vortex is homentropic, initially located at $(x,y)=(0,0)$ and advected downstream by the freestream velocity $V_{\infty}$. The exact solution is given by
\begin{subequations}
\begin{alignat}{1}
\rho & = \rho_{\infty} \Big( 1 - \psi^2 \, M_{\infty} \frac{\gamma-1}{16 \, \pi^2} \, \exp \big( 2 \, \big( 1-r^2 / L^2 \big) \big) \Big)^{1/(\gamma-1)} , \\
V_x & = V_{\infty} \Big( 1 - \psi \, \frac{y}{2 \pi L} \exp \big( 1- r^2 / L^2 \big) \Big) , \\
V_y & = V_{\infty} \, \psi \, \frac{x}{2 \pi L} \exp \big( 1- r^2 / L^2 \big) , \\
p & = \frac{\rho_{\infty}^{1-\gamma} \, V_{\infty}^2}{\gamma \, M_{\infty}^2} \, \rho^{\gamma} , 
\end{alignat}
\end{subequations}
where $r = \big( (x - V_{\infty} \, t)^2 + y^2 \big)^{1/2}$ denotes the distance to the vortex center, $\psi$ is the vortex strength, and $\rho_{\infty}$ and $M_{\infty}$ are the freestream density and Mach number. In particular, we set $\psi = 5$ and $M_{\infty} = \gamma^{-0.5}$, with $\gamma = 1.4$ as discussed before. This completes the non-dimensional description of the problem.

The Shu vortex is simulated in a doubly periodic square $\Omega = (-5 L , 5 L)^2$. The computational domain is partitioned into a triangular mesh obtained by splitting a uniform $10 \times 10$ Cartesian grid into $200$ triangles. Conservation-variable and entropy-variable HDG schemes are considered, and polynomials of degree $k=4$ are used to approximate the solution. The time-step size is $\Delta t = 0.05 \, L \, V_{\infty}^{-1}$ and the simulation is performed from $t_0 = 0$ to $t_f = 500 \, L \, V_{\infty}^{-1}$.

\subsubsection{Numerical results}

Figure \ref{f:timeEvolutions} shows the time evolution of the $L^2$ norm of the error in non-dimensional conservation variables and non-dimensional entropy variables, where $\rho_{\infty}$, $\rho_{\infty} V_{\infty}$ and $\rho_{\infty} V_{\infty}^2$ are used as reference density, momentum and total energy for the non-dimensionalization. 
The evolution of total thermodynamic entropy in the computational domain is shown at the bottom of the figure. Conservation-variable HDG breaks down at time $t = 74.9 \, L \, V_{\infty}^{-1}$, whereas entropy-variable HDG remains stable throughout the simulation. 
The total thermodynamic entropy is non-decreasing in entropy-variable HDG and non-increasing in conservation-variable HDG (i.e.\ the total generalized entropy is non-increasing and non-decreasing, respectively). That is, under-resolution in this problem makes the discretized system evolve through entropy-satisfying and entropy-violating paths for entropy and conservation variables, respectively; which provides critical insights on the mechanisms responsible for numerical instability in the latter case. Induced by the entropy-violating evolution, the error in the conservation-variable solution is much larger when measured in the $\bm{v}$-norm than in the $\bm{u}$-norm.

 \begin{figure}[t!]
 \centering
 {\includegraphics[width=0.65\textwidth]{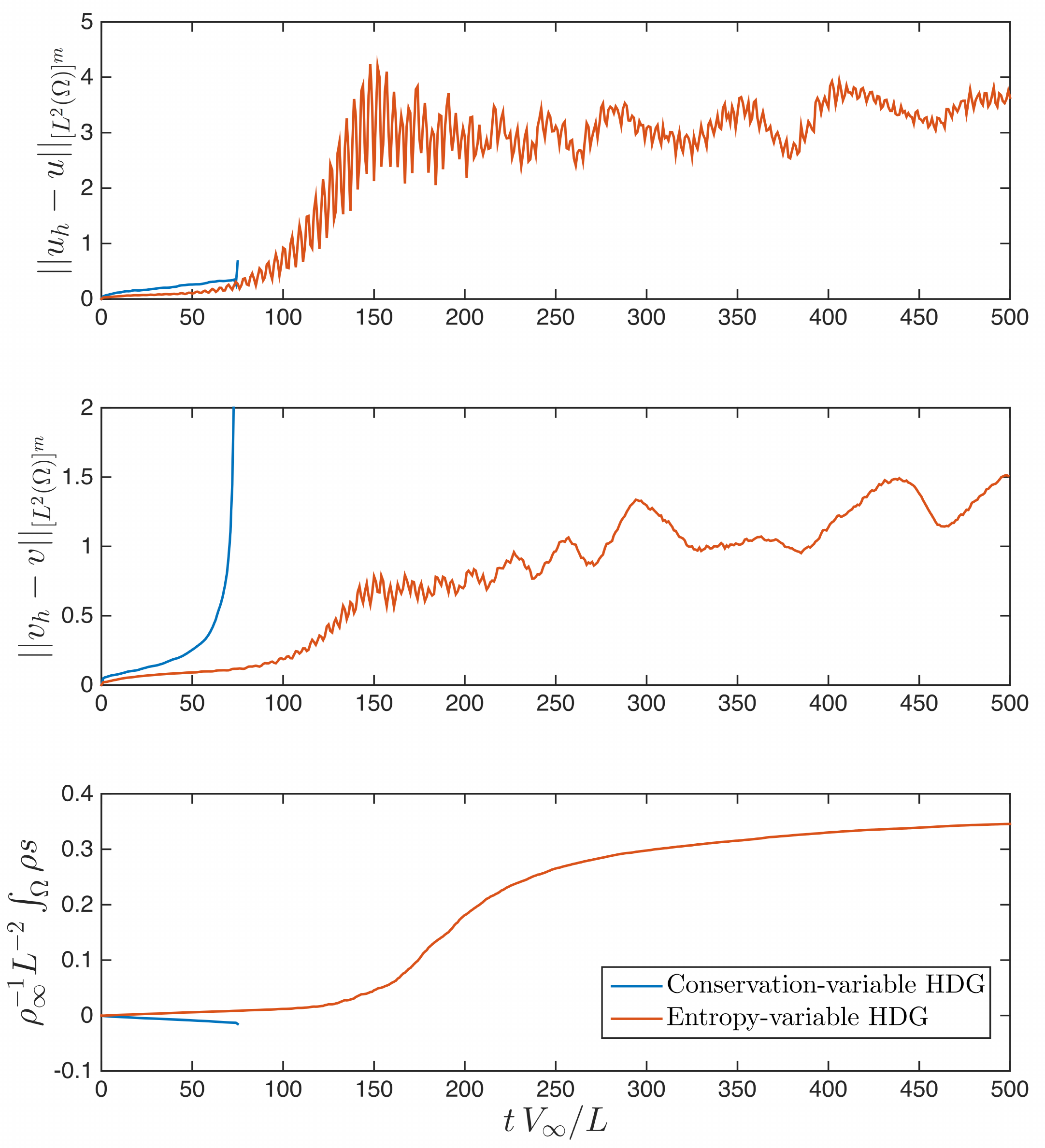}}
  \caption{
Time evolution of the $L^2$ norm of the error in non-dimensional conservation variables (top) and entropy variables (center), as well as the evolution of total thermodynamic entropy in the computational domain (bottom), for the Shu vortex. Conservation-variable and entropy-variable HDG schemes are considered. The subscript $h$ in the error norms denotes the numerical solution and no subscript denotes the exact solution. The initial entropy is used as baseline entropy $s_0$. We recall $s$ is already in non-dimensional form.}\label{f:timeEvolutions}
 \end{figure}

\subsection{\label{s:TGV}Inviscid compressible Taylor-Green vortex}

\subsubsection{Case description and numerical discretization}

The goal of this test case is to examine the stability of the scheme for severely under-resolved compressible flow simulations. To this end, we perform implicit large-eddy simulation (ILES) of the inviscid compressible Taylor-Green vortex (TGV) \cite{Taylor:37}. The TGV problem describes the evolution of the flow in a cubic domain $\Omega = (-L \pi, L \pi)^3$ with triple periodic boundaries, starting from the smooth initial condition
\begin{equation}
\label{initialCondTGV}
\begin{split}
\rho & = \rho_0 , \\
V_x & = V_0 \sin \Big( \frac{x}{L} \Big) \cos \Big( \frac{y}{L} \Big) \cos \Big( \frac{z}{L} \Big) , \\
V_y & = - V_0 \cos \Big( \frac{x}{L} \Big) \sin \Big( \frac{y}{L} \Big) \cos \Big( \frac{z}{L} \Big) , \\
V_z & = 0 , \\
p & = P_0 + \frac{\rho_0 \, V_0^2}{16} \ \bigg( \cos \Big( \frac{2x}{L} \Big) + \cos \Big( \frac{2y}{L} \Big) \bigg) \ \bigg( \cos \Big( \frac{2z}{L} \Big) + 2 \bigg) , 
\end{split}
\end{equation}
where $\rho_0$, $V_0$ and $P_0$ are positive constants. The large-scale eddy in the initial condition leads to smaller and smaller structures through vortex stretching, until the vortical structures eventually break down and the flow transitions to turbulence at $t \approx 8 \, L / V_0$. 
Due to the lack of viscous dissipation, the smallest turbulent length and time scales become arbitrarily small as time evolves. In particular, we consider the inviscid TGV at reference Mach number $M_0 = V_0 / c_0 = 0.8$, where $c_0$ is the speed of sound at temperature $T_0 = P_0 / (\gamma - 1) \, c_v \, \rho_0$. This completes the non-dimensional description of the problem.

For this problem, we consider conservation-variable and entropy-variable EDG schemes. 
In both cases, the computational domain is partitioned into a uniform $32 \times 32 \times 32$ Cartesian grid, and polynomials of degree $k=2$ are used to approximate the solution. The time-step size is set to $\Delta t = 5.71 \cdot 10^{-2} \, L \, V_0^{-1}$ and the numerical solution is computed from $t_0 = 0$ to $t_f = 40 \, L \, V_0^{-1}$. 
Three different phases can be distinguished in the simulation. Before $t \approx 3.5 \, L / V_0$, the flow is laminar and with no subgrid scales. This is followed by an under-resolved laminar phase (i.e.\ subgrid scales appear in the flow) that lasts until $t \approx 8 \, L / V_0$. From then on, the flow is turbulent and under-resolved.

\subsubsection{Numerical results}

Figure \ref{f:TGV_timeEvolutions_WVstudy} shows the temporal evolution of the mean thermodynamic entropy, mean-square vorticity, variance of temperature, and variance of dilatation from $t = 0$ to $t = 10 \, L \, V_0^{-1}$. The vorticity and dilatation are defined as $\bm{\omega} = \nabla \times \bm{V}$ and $\theta = \nabla \cdot \bm{V}$, respectively, and $\langle \, \cdot \, \rangle$ is used to denote spatial averaging. While $\langle \theta \rangle = 0$ in the exact solution due to periodicity in all directions, we note this does not hold exactly, and thus variance of dilatation and mean-square dilatation are different $\langle \theta' \, \theta' \rangle \neq \langle \theta \, \theta \rangle$, in the discrete solution. 
The conservation-variable EDG scheme is unstable and breaks down at $t \approx 7.12 \, L \, V_0^{-1}$, whereas entropy-variable EDG is stable throughout the simulation (from $t = 10 \, L \, V_0^{-1}$ onwards not shown here). 
From the results in Figure \ref{f:TGV_timeEvolutions_WVstudy}, we emphasize that:
\begin{enumerate}
\item When the exact solution does not contain subgrid scales and the simulation is well resolved, that is, before $t \approx 3.5 \, L / V_0$, the conservation-variable and entropy-variable solutions agree well with each other. As subgrid scales appear in the flow and the simulation becomes under-resolved, both numerical solutions start to differ from each other.

\item Before $t \approx 3.5 \, L / V_0$, the total entropy remains approximately constant in both simulations. The conservation-variable and entropy-variable schemes therefore succeed to detect there are no subgrid scales and do not numerically affect entropy under those conditions.

\item The entropy production is non-negligible when there are subgrid scales in the flow. In the entropy-variable scheme, under-resolution leads to an increase in total thermodynamic entropy (decrease in total generalized entropy), and in particular the entropy production is an increasing function of the energy contained in the subgrid scales, that we recall increases over time. In the conservation-variable scheme, however, under-resolution does not lead to an increase entropy. In particular, the numerical oscillations due to under-resolution, which are present in both schemes, are {\it less physical} (in the sense of being less consistent with the Second Law of Thermodynamics) with conservation variables than with entropy variables, and this eventually leads to the breakdown of the simulation.
\end{enumerate}

\begin{figure}[t!]
\centering
\includegraphics[width=1.0\textwidth]{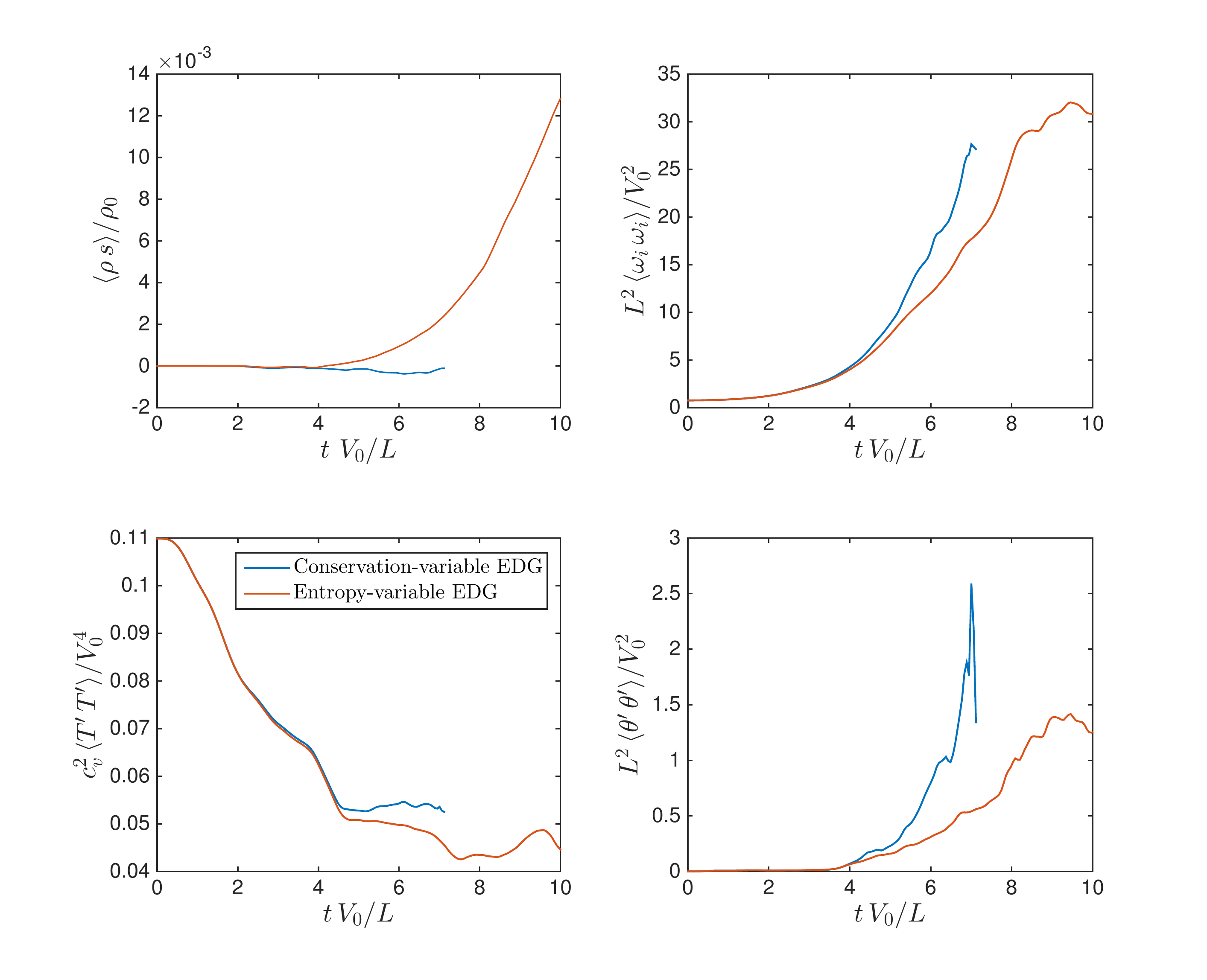}
\caption{\label{f:TGV_timeEvolutions_WVstudy} Temporal evolution of mean thermodynamic entropy, mean-square vorticity, temperature variance, and dilatation variance for the inviscid compressible Taylor-Green vortex. $\langle \, \cdot \, \rangle$ denotes spatial averaging. The mean initial entropy is used as baseline entropy $s_0$. We recall $s$ is already in non-dimensional form.}
\end{figure}

These three observations are justified as follows: On the one hand, if the exact solution does not contain subgrid scales, it is well represented in the approximation space and the scheme is in the asymptotic convergence regime. This implies the inter-element jumps in the numerical solution are small and in particular of order $\norm{ \llbracket \bm{v}_{h} \rrbracket } = \mathcal{O}(h^{k+1})$, where $h$ denotes the element size and $\llbracket \bm{v}_{h} \rrbracket = \bm{v}_{h}^{+} - \bm{v}_{h}^{-}$ is the inter-element jump. Since the amount of entropy introduced 
by entropy-variable schemes is of order $\mathcal{O}(\norm{ \llbracket \bm{v}_{h} \rrbracket }^2)$ (see Proposition \ref{semiDiscGlobEntrIdentity}), it is therefore negligible when the simulation is well resolved. This result applies also to conservation-variable schemes (see footnote \footnote{The amount of thermodynamic entropy introduced by conservation-variable schemes can be shown to be $\mathcal{O}(\norm{ \llbracket \bm{u}_{h} \rrbracket }^2)$ \cite{Fernandez:PhD:2018}. The key difference is that it cannot be shown to be positive in this case.}). On the other hand, if the exact solution contains subgrid scales (i.e.\ if the simulation is under-resolved), the inter-element jumps grow \cite{Fernandez:AIAA:17a,Fernandez:JCP_SGS:18c} and entropy-stable schemes lead to an increase in thermodynamic entropy, as given by Eq. \eqref{entrEqhDG}. Since the amount of entropy introduced by conservation-variable schemes cannot be shown to be positive, this property is not necessarily preserved with conservation variables. 
In fact, since arbitrary oscillations in density and pressure usually reduce the total thermodynamic entropy, it is not completely surprising that DG schemes that are not entropy stable decrease entropy in under-resolved simulations of compressible flows, in which density and pressure oscillate inside the elements.

Finally, we note that the numerical scheme increasing entropy in under-resolved computations is not only important for stability purposes, but also provides with a built-in (implicit) ``subgrid-scale model'' for large-eddy simulation of compressible flows.


\subsection{\label{s:CIT}Decay of compressible, homogeneous, isotropic turbulence}


\subsubsection{Case description and numerical discretization}

We perform implicit large-eddy simulation of the decay of compressible, homogeneous, isotropic turbulence with eddy shocklets \cite{Lee:1991}. The goal of this test problem is to investigate the stability and accuracy of the scheme for under-resolved simulations of compressible turbulence. The problem domain consists of a cube $\Omega = (-L \pi , L \pi)^3$ with triple periodic boundaries. The initial density, pressure and temperature fields are constant, and the initial velocity is solenoidal and with kinetic energy spectrum satisfying $E(k) \sim k^4 \exp [ - 2 \, (k / k_M)^2 ]$, where $k_M$ corresponds to the most energetic wavenumber and is set to $k_M = 4 / L$. The details of the procedure to generate the initial velocity field are described in \cite{Johnsen:2010}. The initial turbulent Mach number and Taylor-scale Reynolds number are
$$ \qquad M_{t,0} := \frac{\sqrt{ \langle V_{i,0} \, V_{i,0} \rangle}}{\langle c_0 \rangle} = 0.6 , \qquad \qquad Re_{\lambda,0} := \frac{\langle \rho_0 \rangle \, V_{rms,0} \, \lambda_0}{\langle \mu_0 \rangle} = 100 , $$
where the zero subscript denotes the initial value, $\langle \, \cdot \, \rangle$ denotes spatial averaging, and
$$ V_{rms} := \sqrt{\frac{\langle V_i \, V_i \rangle}{3}} , \qquad \qquad \lambda := \sqrt{\frac{ \langle V_1^2 \rangle }{ \langle (\partial_1 V_1)^2 \rangle}} $$
are the root mean square velocity and the Taylor microscale, respectively. Also, the dynamic viscosity is assumed to follow a power-law of the form
\begin{equation}
\mu = \mu_0 \, \bigg( \frac{T}{T_0} \bigg)^{3/4} . 
\end{equation}
This completes the non-dimensional description of the problem. Due to the imbalance in the initial condition, strong vortical, entropy and acoustic modes (i.e.\ all the compressible modes) develop and persist throughout the simulation. Weak shock waves (eddy shocklets) appear spontaneously from the turbulent motions as well.

The computational domain is discretized into a uniform $32 \times 32 \times 32$ Cartesian grid and polynomials of degree $k=2$ are used to approximate the solution; which leads to severe spatial under-resolution for this problem \cite{Hillewaert:2016,Johnsen:2010}. The simulation is performed from $t_0 = 0$ to $t_f = 4 \, \tau_0$ with time-step size $\Delta t = 1.183 \cdot 10^{-2} \, \tau_0$, where $\tau_0 = \lambda_0 / V_{rms,0}$ denotes the initial eddy turn-over time. This corresponds to a CFL number based on the initial mean-square velocity of $V_{rms,0} \, \Delta t / h = 0.02$. Conservation-variable and entropy-variable EDG schemes 
are considered.

\subsubsection{Numerical results}

Figure \ref{f:timeEvolutionsCIT_WVstudy} shows the temporal evolution of the mean-square velocity, mean-square vorticity, temperature variance, and dilatation variance for conservation-variable EDG, entropy-variable EDG, and the direct numerical simulation (DNS) data from Hillewaert {\it et al.} \cite{Hillewaert:2016}. 
The conservation-variable EDG scheme is unstable and breaks down at $t \approx 0.450 \, \tau_0$, whereas entropy-variable EDG is stable throughout the simulation. 
In addition, the entropy-variable scheme shows very good agreement with the DNS data, particularly when compared to the LES results obtained with other numerical schemes \cite{Hillewaert:2016,Johnsen:2010} and despite a slightly higher resolution was used in the simulations therein. The main discrepancy with DNS is observed for the time evolution of dilatation variance. We note, however, that the grid resolution $\hbar$ in DNS was 
such that the cell P\'eclet number $Pe_{\hbar,0} := \langle \rho_0 \rangle \, v_{rms,0} \, \hbar / \langle \mu_0 \rangle$ is approximately $3.3$. 
While this suffices to stabilize the shock waves, it may not suffice to accurately resolve them and it is therefore unclear whether the DNS results are grid converged. Some differences between unfiltered DNS solutions computed with a finite-volume code and a DG code are indeed reported in \cite{Hillewaert:2016}. 
In short, Figure \ref{f:timeEvolutionsCIT_WVstudy} shows the use of entropy variables stabilizes the scheme while having a small impact on the propagation of vortical, entropy and acoustic modes; which is critical for large-eddy simulation of compressible flows.

Figure \ref{f:numDissEntr_cHIT} shows the time evolution of total thermodynamic entropy in the computational domain as well as of the quantity
\begin{equation}
\label{e:PiS}
\Pi_{S} := \frac{d}{dt} \int_{\Omega} \rho s - \frac{1}{c_v} \int_{\Omega} \bigg( \frac{\Phi}{T} + \kappa \frac{\norm{\nabla T}^2}{T^2} \bigg) , 
\end{equation}
where $\Phi = \nabla \bm{V} : \bm{\tau}$ is the viscous dissipation of kinetic energy and $:$ is the Frobenius inner product of two matrices. 
The second term in the right-hand side of Eq. \eqref{e:PiS} 
corresponds to the generation of thermodynamic entropy due to physical mechanisms and is non-positive provided $\mu , \beta , \kappa \geq 0$. $\Pi_S$ is therefore the contribution to the entropy production due to the numerical scheme, and is referred to as the {\it numerical generation of entropy} (if positive) or {\it numerical destruction of entropy} (if negative). We note that $\Phi$ and $\nabla T$ can be computed using either $(\bm{v}_h , \nabla \bm{v}_h)$ or $(\bm{v}_h , \bm{q}_h)$. Since $\nabla \bm{v}_h$ converges suboptimally, whereas $\bm{q}_h$ converges optimally for some schemes within the hybridized DG family, as discussed in Section \ref{s:couette}, the latter approach is adopted here. From this figure, the total thermodynamic entropy in the domain increases over time (total generalized entropy decreases) both with conservation and entropy variables. Conservation variables, however, lead to very large numerical entropy destruction and this is in turn responsible for the simulation breakdown. This behavior is not observed with entropy variables. 

{\bf Remark:} For the Euler equations, $\Pi_S$ vanishes and entropy stability implies $\Pi_S$ is non-negative. For the Navier-Stokes equations, however, entropy stability is not a sufficient condition for $\Pi_S \geq 0$. In particular, it follows from the entropy stability proof for the Navier-Stokes equations in \ref{s:appProof} that non-negativity of $\Pi_S$ requires both entropy stability and
\begin{equation}
\label{e:tooBadWeDontHaveIt}
\int_{\mathcal{T}_h} \bm{q}_h^t \cdot \widetilde{\bm{\mathcal{K}}} (\bm{v}_h) \cdot \bm{q}_h \geq \frac{1}{c_v} \int_{\Omega} \bigg( \frac{\Phi}{T} + \kappa \frac{\norm{\nabla T}^2}{T^2} \bigg) . 
\end{equation}
While
\begin{equation}
\label{e:identity_}
\nabla \bm{v}^t \cdot \widetilde{\bm{\mathcal{K}}} (\bm{v}) \cdot \nabla \bm{v} = \frac{1}{c_v} \bigg( \frac{\Phi}{T} + \kappa \frac{\norm{\nabla T}^2}{T^2} \bigg) 
\end{equation}
holds pointwise as an identity for any physical and differentiable $\bm{v}$ field, Eq. \eqref{e:tooBadWeDontHaveIt} cannot be ensured regardless of whether $(\bm{v}_h , \nabla \bm{v}_h)$ or $(\bm{v}_h , \bm{q}_h)$ are used to compute $\Phi$ and $\nabla T$ (see footnote 
\footnote{For non-curved elements (i.e.\ $p=1$), it can be shown \cite{Fernandez:PhD:2018} that $\bm{q}_h = \nabla \bm{v}_h + \bm{\mathcal{R}}_{\partial \mathcal{T}_h} ( \bm{v}_h|_{\partial \mathcal{T}_h} - \widehat{\bm{v}}_h ; \bm{n})$, where $\bm{\mathcal{R}}_{\partial \mathcal{T}_h} : \bm{\mathcal{V}}_{h}^k|_{\partial \mathcal{T}_h} \to \bm{\mathcal{Q}}_h^k \ ; \ \bm{a} \mapsto \bm{\mathcal{R}}_{\partial \mathcal{T}_h}(\bm{a})$ is a lifting operator such that
$$ \int_{\partial \mathcal{T}_h} \bm{b}_{ij} \bm{a}_i \bm{n}_j = - \int_{\mathcal{T}_h} \bm{b}_{ij} [\bm{\mathcal{R}}_{\partial \mathcal{T}_h}(\bm{a})]_{ij} , \qquad \forall \bm{b} \in \bm{\mathcal{Q}}_h^k , $$
with $i = 1 , \dots , m$ and $j = 1 , \dots , d$. The existence and uniqueness of $\bm{\mathcal{R}}_{\partial \mathcal{T}_h}$ is given by the Riesz representation theorem. Using this result and Eq. \eqref{e:identity_}, it follows that
$$ \int_{\mathcal{T}_h} \bm{q}_h^t \cdot \widetilde{\bm{\mathcal{K}}} (\bm{v}_h) \cdot \bm{q}_h = \frac{1}{c_v} \int_{\Omega} \bigg( \frac{\Phi}{T} + \kappa \frac{\norm{\nabla T}^2}{T^2} \bigg) \bigg|_{(\bm{v}_h , \nabla \bm{v}_h)} + 2 \int_{\mathcal{T}_h} \bm{\mathcal{R}}_{\partial \mathcal{T}_h}^t \cdot \widetilde{\bm{\mathcal{K}}} (\bm{v}_h) \cdot \nabla \bm{v}_h + \int_{\mathcal{T}_h} \bm{\mathcal{R}}_{\partial \mathcal{T}_h}^t \cdot \widetilde{\bm{\mathcal{K}}} (\bm{v}_h) \cdot \bm{\mathcal{R}}_{\partial \mathcal{T}_h} . $$
This equality holds true and is to be compared to the desired inequality in Eq. \eqref{e:tooBadWeDontHaveIt}. For curved elements (i.e.\ $p \geq 2$), a similar result applies.
} for further discussion).

 \begin{figure}
 \centering
 {\includegraphics[width=1.0\textwidth]{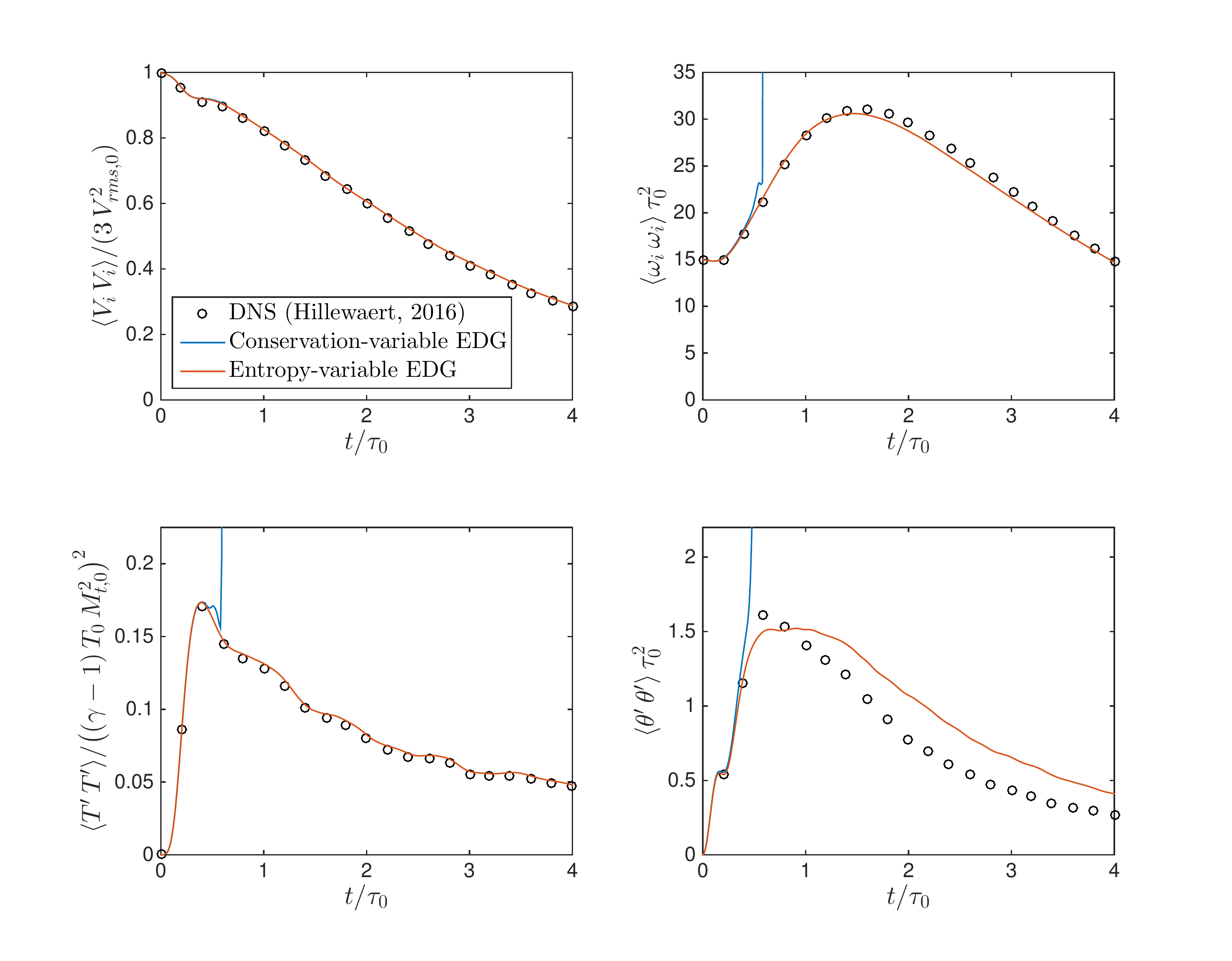}}
  \caption{Temporal evolution of mean-square velocity (top left), mean-square vorticity (top right), temperature variance (bottom left), and dilatation variance (bottom right) for the decay of compressible, homogeneous, isotropic turbulence. The zero subscript denotes the initial value and $\langle \, \cdot \, \rangle$ denotes spatial averaging.}\label{f:timeEvolutionsCIT_WVstudy}
 \end{figure}
 
  \begin{figure}
 \centering
 {\includegraphics[width=0.49\textwidth]{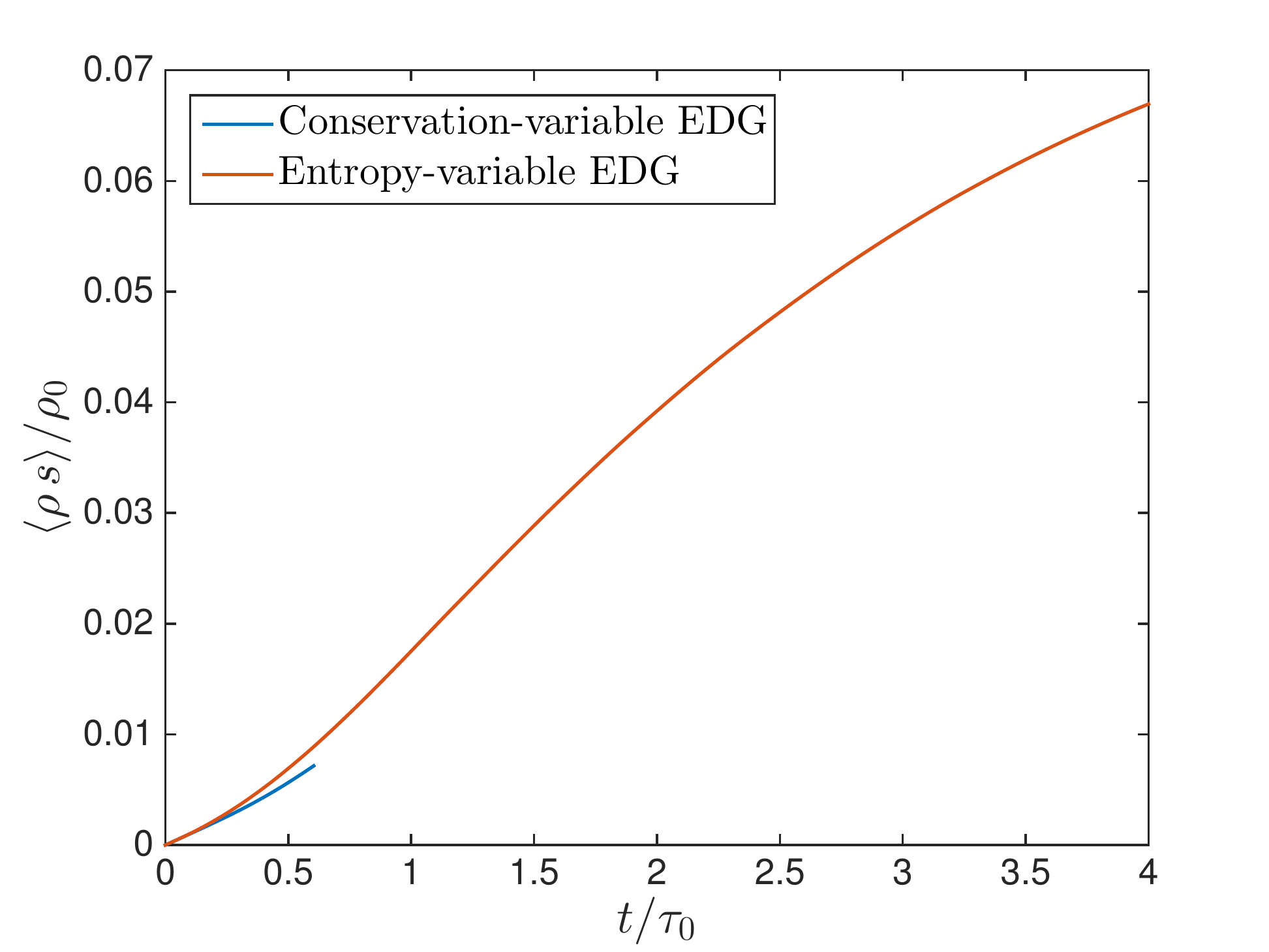}}
 \hfill {\includegraphics[width=0.49\textwidth]{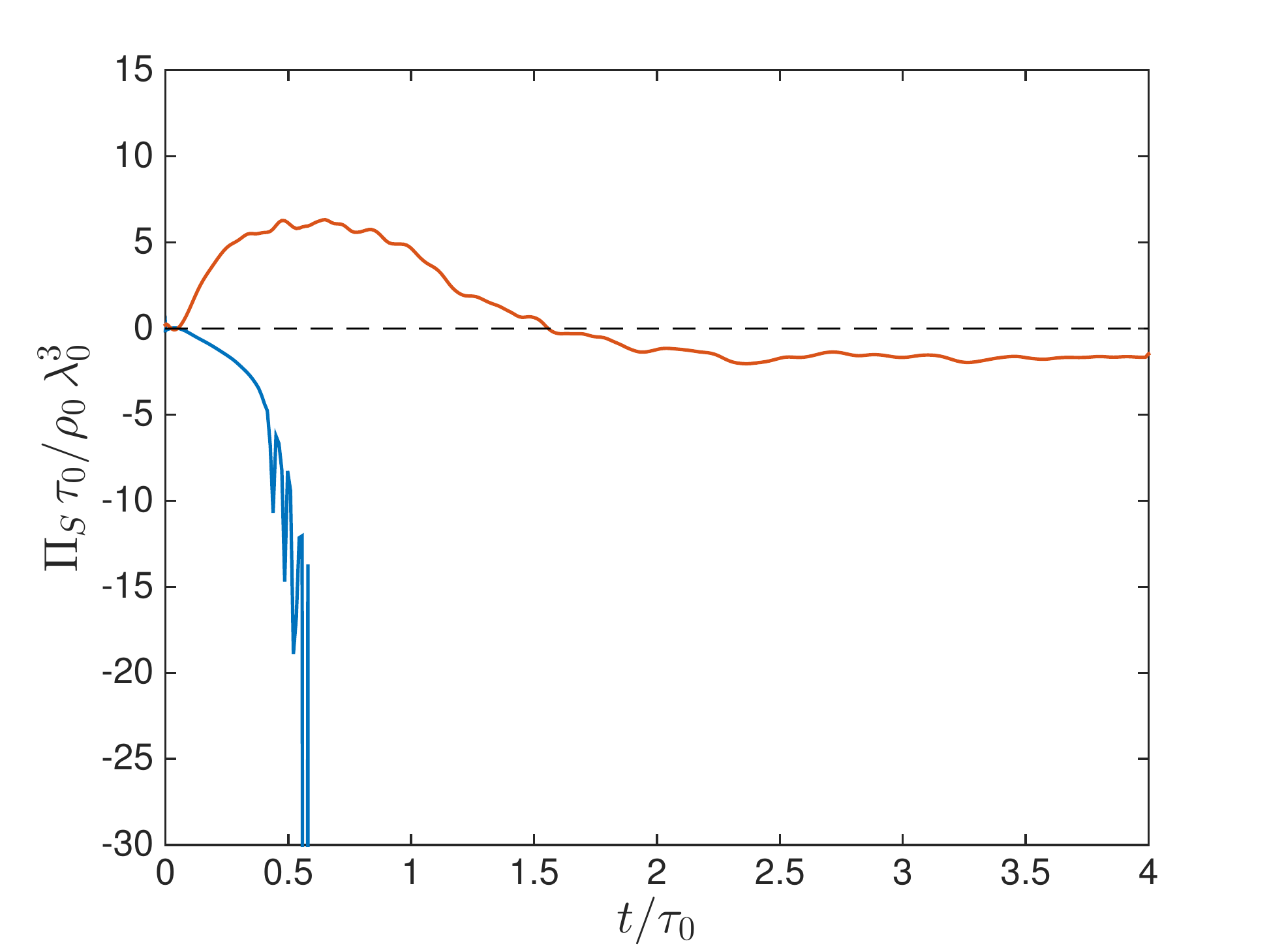}}
  \caption{Temporal evolution of mean thermodynamic entropy (left) and numerical generation of entropy as defined in Eq. \eqref{e:PiS} (right) for the decay of compressible, homogeneous, isotropic turbulence. $\langle \, \cdot \, \rangle$ denotes spatial averaging. The mean initial entropy is used as baseline entropy $s_0$. We recall $s$ is already in non-dimensional form.}\label{f:numDissEntr_cHIT}
 \end{figure}

\section{\label{s:conclusions}Conclusions}

We presented the entropy-variable hybridized DG methods for the compressible Euler and Navier-Stokes equations. The proposed schemes display optimal accuracy order and have important advantages over existing DG methods. First, as hybridized DG methods, they result in significantly fewer globally coupled degrees of freedom and number of nonzero entries in the Jacobian matrix than in standard DG methods; which allows for more computationally efficient implementations, both in terms of flop count and memory requirements. Second, as entropy-stable schemes, they lead to improved robustness and improved accuracy in under-resolved simulations of compressible flows with respect to their conservation-variable counterparts, as illustrated through a number of steady and unsteady flow problems in subsonic, transonic, and supersonic regimes.



\section*{Acknowledgments}

The authors acknowledge the Air Force Office of Scientific Research (FA9550-16-1-0214), the National Aeronautics and Space Administration (NASA NNX16AP15A) and Pratt \& Whitney for supporting this effort. The first author also acknowledges the financial support from the Zakhartchenko and ``la Caixa'' Fellowships.

\appendix

\section{\label{s:appProof}Proof of Proposition \ref{semiDiscGlobEntrIdentity}, Theorem \ref{mainTh}, Proposition \ref{semiDiscGlobEntrIdentityNS} and Theorem \ref{mainThNS}}

\subsection{Supporting lemmas}

\begin{lemma}[Jump entropy identity]
\label{lemmaInterElemJump}
For any $\bm{n} \in \mathbb{R}^d$ and any pair of physical states $\bm{v}_1, \bm{v}_2 \in X_v$, the following identity holds:
\begin{equation}
\begin{split}
\label{e:lemmaInterElemJump}
- & \big[ \bm{\mathcal{F}}_{n} \big] _{\bm{v}_{1}}^{\bm{v}_{2}} + \frac{1}{2} \big( \bm{v}_{1} + \bm{v}_{2} \big)^t \cdot \big[ \bm{F}_n \big]_{\bm{v}_1}^{\bm{v}_2} = \frac{1}{2} \big( \bm{v}_2 - \bm{v}_1 \big)^t \cdot \bm{\Sigma} (\bm{v}_1 , \bm{v}_2 ; \bm{n} ) \cdot \big( \bm{v}_2 - \bm{v}_1 \big) . 
\end{split}
\end{equation}
We recall that $\bm{\Sigma}(\bm{v}_1 , \bm{v}_2 ; \bm{n})$ was defined in Eq. \eqref{e:Lambda_n} as
\begin{equation}\tag{\ref{e:Lambda_n}}
\begin{split}
\bm{\Sigma} (\bm{v}_1 , \bm{v}_2 ; \bm{n}) := & \int_0^1 (1 - \theta) \, \Big( \tilde{\bm{A}}_n \big(\tilde{\bm{v}}(\theta ; \bm{v}_1 , \bm{v}_2) ; \bm{n} \big) - \tilde{\bm{A}}_n \big(\tilde{\bm{v}}(\theta ; \bm{v}_2 , \bm{v}_1) ; \bm{n} \big) \Big) \, d \theta \\
= & \int_0^1 (1 - 2 \theta) \ \tilde{\bm{A}}_n \big(\tilde{\bm{v}}(\theta ; \bm{v}_1 , \bm{v}_2) ; \bm{n} \big) \, d \theta = - \int_0^1 (1 - 2 \theta) \ \tilde{\bm{A}}_n \big(\tilde{\bm{v}}(\theta ; \bm{v}_2 , \bm{v}_1) ; \bm{n} \big) \, d \theta . 
\end{split}
\end{equation}
\end{lemma}
\begin{proof} 
Equation \eqref{e:lemmaInterElemJump} is a trivial generalization of Lemma 11 in \cite{Barth:99}. The last two equalities in Eq. \eqref{e:Lambda_n} follow from the change of variable $\theta' = 1 - \theta$ applied to the second and first terms in the integrand, respectively.
\end{proof}

\begin{corol}[Pointwise entropy production due to face jumps]
\label{corolLemmaInterElemJump} 
Let $F \in \partial \mathcal{T}_h$ be a face, $\bm{n} \in \mathbb{R}^d$ be the unit normal vector to $F$ pointing outwards from the element, and $\widehat{\bm{v}}_h (\bm{x}), \bm{v}_h (\bm{x}) \in X_v$ be physical states for all $\bm{x}$ in $F$. The following identity holds pointwise on $F$:
\begin{equation}
\begin{split}
\label{e:corolLemmaInterElemJump3}
- & \big[ \bm{\mathcal{F}}_{n} \big] _{\widehat{\bm{v}}_h}^{\bm{v}_h} + \frac{1}{2} \big( \widehat{\bm{v}}_h + \bm{v}_h \big)^t \cdot \big[ \bm{F}_{n} \big]_{\widehat{\bm{v}}_h}^{\bm{v}_h} = \frac{1}{2} \big( \bm{v}_h - \widehat{\bm{v}}_h \big)^t \cdot \bm{\Sigma} (\widehat{\bm{v}}_h , \bm{v}_h ; \bm{n} ) \cdot \big( \bm{v}_h - \widehat{\bm{v}}_h \big) . 
\end{split}
\end{equation}
\end{corol}



\subsection{Proof of Proposition \ref{semiDiscGlobEntrIdentity}}


Let $(\bm{v}_h(t),\widehat{\bm{v}}_h(t))$ denote the numerical solution at time $t > 0$, and let $a_h(\bm{v}_h,\widehat{\bm{v}}_h,\partial\bm{v}_h / \partial t; \bm{w})$ and $b_h(\bm{v}_h,\widehat{\bm{v}}_h;\bm{\mu})$ be shorthand notations for the left-hand sides in Equations \eqref{e:hDG1} and \eqref{e:hDG2}, respectively. Note we have omitted, and will omit hereinafter, the time dependency of the solution to simplify the notation. Integrating the inviscid flux term in $a_h$ by parts\footnote{We can integrate by parts since, first, $\mathcal{T}_h$ is Lipschitz and $\bm{\mathcal{V}}_h^k \subset \mathcal{C}^{\infty} (\mathcal{T}_h ; \mathbb{R}^m)$ for non-singular elements, and, second, $\bm{F} \in \mathcal{C}^{\infty} (X_v ; \mathbb{R}^{m \times d})$ for physical states $\bm{v} \in X_v \implies \bm{F}(\bm{v}_h) \in \mathcal{C}^{\infty} (\mathcal{T}_h  ; \mathbb{R}^{m \times d})$ for physical solutions and non-singular elements.}, these maps can be rewritten as
\begin{subequations}
\begin{alignat}{2}
a_h(\bm{v}_h,\widehat{\bm{v}}_h,\partial \bm{v}_h / \partial t;\bm{w}) & = \int_{\mathcal{T}_h} \bm{w}^t \cdot \frac{\partial \bm{u}(\bm{v}_h)}{\partial t} + \int_{\mathcal{T}_h} \bm{w}^t \cdot \big( \nabla \cdot \bm{F} (\bm{v}_h) \big) + \int_{\partial \mathcal{T}_h} \bm{w}^t \cdot \big( \widehat{\bm{f}}_h - \bm{F}_n (\bm{v}_h) \big) , \\
b_h(\bm{v}_h,\widehat{\bm{v}}_h;\bm{\mu}) & = \int_{\partial \mathcal{T}_h} \bm{\mu}^t \cdot \widehat{\bm{f}}_h - \int_{\partial \Omega} \bm{\mu}^t \cdot \big( \widehat{\bm{f}}_h - \widehat{\bm{b}}_h (\widehat{\bm{v}}_h,\bm{v}_h ; \bm{v}^{\partial \Omega}) \big) . 
\end{alignat}
\end{subequations}
Since ({\em i}) Equations \eqref{e:hDG1} and \eqref{e:hDG2} hold for all $(\bm{w},\bm{\mu}) \in \bm{\mathcal{V}}_h^k \otimes \bm{\mathcal{M}}_h^k$, and ({\em ii}) the approximation and test spaces are the same, it follows that $a_h(\bm{v}_h,\widehat{\bm{v}}_h,\partial \bm{v}_h / \partial t;\bm{v}_h) - b_h(\bm{v}_h,\widehat{\bm{v}}_h;\widehat{\bm{v}}_h) = 0$, that is,
\begin{equation}
\label{e:tmp1}
\begin{split}
0 & = \int_{\mathcal{T}_h} \bm{v}_h^t \cdot \frac{\partial \bm{u}(\bm{v}_h)}{\partial t} + \int_{\mathcal{T}_h} \bm{v}_h^t \cdot \big( \nabla \cdot \bm{F} (\bm{v}_h) \big) + \int_{\partial \mathcal{T}_h} \bm{v}_h^t \cdot \big( \widehat{\bm{f}}_h - \bm{F}_n (\bm{v}_h) \big) - \int_{\partial \mathcal{T}_h} \widehat{\bm{v}}_h^t \cdot \widehat{\bm{f}}_h \\
& + \int_{\partial \Omega} \widehat{\bm{v}}_h^t \cdot \big( \widehat{\bm{f}}_h - \widehat{\bm{b}}_h (\widehat{\bm{v}}_h,\bm{v}_h ; \bm{v}^{\partial \Omega}) \big) \\
& = \frac{d}{dt} \int_{\mathcal{T}_h} H(\bm{v}_h) + \int_{\mathcal{T}_h} \nabla \cdot \bm{\mathcal{F}} (\bm{v}_h) + \int_{\partial \mathcal{T}_h} \frac{1}{2} \, \bm{v}_h^t \cdot \big( \bm{F}_n(\widehat{\bm{v}}_h) - \bm{F}_n(\bm{v}_h) \big) + \frac{1}{2} \int_{\partial \mathcal{T}_h} \bm{v}_h^t \cdot \bm{{\sigma}}(\widehat{\bm{v}}_h , \bm{v}_h ; \bm{n}) \cdot \big( \bm{v}_h - \widehat{\bm{v}}_h \big) \\
& - \int_{\partial \mathcal{T}_h} \frac{1}{2} \, \widehat{\bm{v}}_h^t \cdot \big( \bm{F}_n(\bm{v}_h) + \bm{F}_n(\widehat{\bm{v}}_h) \big) - \frac{1}{2} \int_{\partial \mathcal{T}_h} \widehat{\bm{v}}_h^t \cdot \bm{{\sigma}}(\widehat{\bm{v}}_h , \bm{v}_h ; \bm{n}) \cdot \big( \bm{v}_h - \widehat{\bm{v}}_h \big) \\
& + \int_{\partial \Omega} \widehat{\bm{v}}_h^t \cdot \big( \widehat{\bm{f}}_h - \widehat{\bm{b}}_h (\widehat{\bm{v}}_h,\bm{v}_h ; \bm{v}^{\partial \Omega}) \big) , 
\end{split}
\end{equation}
where we have used the identities in Eq. \eqref{e:vTidentity}, the definition of the inviscid numerical flux \eqref{numericalFlux}, and the fact that the mesh $\mathcal{T}_h$ is stationary. 
Applying the divergence theorem and noting that
\begin{equation}
\int_{\partial \mathcal{T}_h \backslash \partial \Omega} \bm{\mathcal{F}}_n (\widehat{\bm{v}}_h) = 0 , \qquad \qquad \int_{\partial \mathcal{T}_h \backslash \partial \Omega} \widehat{\bm{v}}_h^t \cdot \bm{F}_n (\widehat{\bm{v}}_h) = 0 
\end{equation}
due to $\pm$ duplication of interior faces in $\partial \mathcal{T}_h \backslash \partial \Omega$, Eq. \eqref{e:tmp1} can be written as
\begin{equation}
\begin{split}
0 & = \frac{d}{dt} \int_{\mathcal{T}_h} H(\bm{v}_h) - \int_{\partial \mathcal{T}_h} \big( \bm{\mathcal{F}}_n (\widehat{\bm{v}}_h) - \bm{\mathcal{F}}_n (\bm{v}_h) \big) + \int_{\partial \mathcal{T}_h} \frac{1}{2} \, \big( \bm{v}_h + \widehat{\bm{v}}_h \big) ^t \cdot \big( \bm{F}_n(\widehat{\bm{v}}_h) - \bm{F}_n(\bm{v}_h) \big) \\
& + \frac{1}{2} \int_{\partial \mathcal{T}_h} \big( \bm{v}_h - \widehat{\bm{v}}_h \big)^t \cdot \bm{{\sigma}}(\widehat{\bm{v}}_h , \bm{v}_h ; \bm{n}) \cdot \big( \bm{v}_h - \widehat{\bm{v}}_h \big) \\
& + \int_{\partial \Omega} \bm{\mathcal{F}}_n (\widehat{\bm{v}}_h) - \int_{\partial \Omega} \widehat{\bm{v}}_h^t \cdot \bm{F}_n(\widehat{\bm{v}}_h) + \int_{\partial \Omega} \widehat{\bm{v}}_h^t \cdot \big( \widehat{\bm{f}}_h - \widehat{\bm{b}}_h (\widehat{\bm{v}}_h,\bm{v}_h ; \bm{v}^{\partial \Omega}) \big) . 
\end{split}
\end{equation}
The desired result then readily follows by applying Corollary \ref{corolLemmaInterElemJump}.

\subsection{\label{s:EulerProof}Proof of Theorem \ref{mainTh}}

Equation \eqref{e:entropyStability} is a corollary of Proposition \ref{semiDiscGlobEntrIdentity}. The upper bound in Eq. \eqref{e:entropyBound} then trivially follows from the Newton-Leibniz formula and positivity of integration. For the lower bound, let us consider the Taylor series with integral remainder of the entropy function between the states $\bm{u}^*(\bm{v}_h)$ and $\bm{u}(\bm{v}_h)$, namely,
\begin{equation}
\label{e:intRem}
H(\bm{u}) = H(\bm{u}^*) + \bigg( \frac{\partial H(\bm{u}^*)}{\partial \bm{u}} \bigg) ^t \cdot (\bm{u} - \bm{u}^*) + \int_0^1 (1-\theta) \, (\bm{u} - \bm{u}^*)^t \cdot \frac{\partial^2 H ( \bm{u}^* + \theta \, (\bm{u}-\bm{u}^*) )}{\partial \bm{u}^2} \cdot (\bm{u} - \bm{u}^*) \, d \theta . 
\end{equation}
Note that, for physical solutions, $H$ is sufficiently regular for Eq. \eqref{e:intRem} to hold everywhere in $\mathcal{T}_h$. Integrating \eqref{e:intRem} over $\mathcal{T}_h$, the second term in the right-hand side vanishes by definition of $\bm{u}^*$. It then follows from convexity of $H$ that
\begin{equation}
\int_{\mathcal{T}_h} H ( \bm{u}^*(\bm{v}_h(t)) ) \leq \int_{\mathcal{T}_h} H ( \bm{u}(\bm{v}_h(t)) ) . 
\end{equation}
The lower bound is finally established by noting that $\bm{u}^*(t)$ is constant in time, and in particular equal to $\bm{u}^*(\bm{v}_{h,0})$, since the entropy-variable hybridized DG methods are $\bm{u}$-conservative by construction, as discussed in Section \ref{s:hDG_euler}.

\subsection{\label{s:NSproof_timeEvol}Proof of Proposition \ref{semiDiscGlobEntrIdentityNS}}

Proposition \ref{semiDiscGlobEntrIdentityNS} is an extension of Proposition \ref{semiDiscGlobEntrIdentity}. In particular, let $c_h(\bm{v}_h,\bm{q}_h,\widehat{\bm{v}}_h,\partial \bm{v}_h / \partial t; \bm{w})$ and $d_h(\bm{v}_h,\bm{q}_h,\widehat{\bm{v}}_h;\bm{\mu})$ be shorthand notations for the left-hand sides in Equations \eqref{e:hDG_NS1} and \eqref{e:hDG_NS2}, respectively, where we have again omitted the time dependency of the solution. Since ({\em i}) Equations \eqref{e:hDG_NS1} and \eqref{e:hDG_NS2} hold for all $(\bm{w},\bm{\mu}) \in \bm{\mathcal{V}}_h^k \otimes \bm{\mathcal{M}}_h^k$, and ({\em ii}) the approximation and test spaces are the same, it follows that $c_h(\bm{v}_h,\bm{q}_h,\widehat{\bm{v}}_h,\partial \bm{v}_h / \partial t; \bm{v}_h) - d_h(\bm{v}_h,\bm{q}_h,\widehat{\bm{v}}_h;\widehat{\bm{v}}_h) = 0$, 
that is,
\begin{equation}
\label{e:eq124}
\begin{split}
0 & = a_h(\bm{v}_h,\widehat{\bm{v}}_h,\partial \bm{v}_h / \partial t;\bm{v}_h) - b_h(\bm{v}_h,\widehat{\bm{v}}_h;\widehat{\bm{v}}_h) \\
& - \int_{\mathcal{T}_h} \nabla \bm{v}_h^t \cdot \bm{G} (\bm{v}_h,\bm{q}_h) + \int_{\partial \mathcal{T}_h} \big( \bm{v}_h - \widehat{\bm{v}}_h \big)^t \cdot \widehat{\bm{g}}_h + \int_{\partial \Omega} \widehat{\bm{v}}_h^t \cdot \big( \widehat{\bm{g}}_h - \widehat{\bm{b}}_h^{\mathcal{V}} (\bm{v}_h,\bm{q}_h,\widehat{\bm{v}}_h ; \bm{v}^{\partial \Omega}) \big) , 
\end{split}
\end{equation}
where $\widehat{\bm{b}}_h^{\mathcal{V}}$ denotes the viscous contribution to the boundary condition term, i.e.\ the viscous terms in $\widehat{\bm{b}}_h$. Furthermore, using the $\bm{\mathcal{Q}}_h^k$ projection of $\bm{G}(\bm{v}_h , \bm{q}_h)$ as test function in Eq. \eqref{e:hDG_NS0}, the following identity holds for non-curved elements (i.e.\ $p=1$)
\begin{equation}
\label{e:eq123}
- \int_{\mathcal{T}_h} \nabla \bm{v}_h^t \cdot \bm{G} (\bm{v}_h,\bm{q}_h) = - \int_{\partial \mathcal{T}_h} \big( \bm{v}_h - \widehat{\bm{v}}_h \big)^t \cdot \Big( \Pi_{\bm{\mathcal{Q}}_h^k} \big[ \bm{G}(\bm{v}_h , \bm{q}_h) \big] \cdot \bm{n} \Big) + \int_{\mathcal{T}_h} \bm{q}_h^t \cdot \widetilde{\bm{\mathcal{K}}} (\bm{v}_h) \cdot \bm{q}_h , 
\end{equation}
where we have used ({\em i}) integration by parts\footnote{We can integrate by parts since, first, $\mathcal{T}_h$ is Lipschitz, $\bm{\mathcal{V}}_h^k \subset \mathcal{C}^{\infty} (\mathcal{T}_h ; \mathbb{R}^m)$ and $\bm{\mathcal{Q}}_h^k \subset \mathcal{C}^{\infty} (\mathcal{T}_h ; \mathbb{R}^{m \times d})$ for non-singular elements, and, second, $\bm{G} \in \mathcal{C}^{\infty} (X_v \otimes \mathbb{R}^{m \times d} ; \mathbb{R}^{m \times d})$ for physical states $\bm{v} \in X_v \implies \bm{G}( \bm{v}_h , \bm{q}_h ) \in \mathcal{C}^{\infty} (\mathcal{T}_h  ; \mathbb{R}^{m \times d})$ for physical solutions and non-singular elements $ \implies \big( \Pi_{\bm{\mathcal{Q}}_h^k} \big[ \bm{G} (\bm{v}_h , \bm{q}_h) \big] \big)  \in \mathcal{C}^{\infty} (\mathcal{T}_h  ; \mathbb{R}^{m \times d})$ for physical solutions and non-singular elements.}, ({\em ii}) $\bm{\mathcal{Q}}_h^k$-orthogonality, 
and ({\em iii}) the fact that $\nabla \bm{v}_h \in \bm{\mathcal{Q}}_h^k$ and $\nabla \bm{n}|_F = 0 \ \forall F \in \partial \mathcal{T}_h$ for non-curved elements. 
Note that $\widetilde{\bm{\mathcal{K}}}_{ij} = \bm{\mathcal{K}}_{ij} \, \tilde{\bm{A}}_0 , \, i,j=1,\dots,d$ are symmetric positive semi-definite \cite{Hughes:86}. 
The desired result then readily follows from Equations \eqref{e:eq124}$-$\eqref{e:eq123} and Proposition \ref{semiDiscGlobEntrIdentity}.

\subsection{\label{s:NSproof}Proof of Theorem \ref{mainThNS}}

Equation \eqref{e:entropyStabilityNS} is a corollary of Proposition \ref{semiDiscGlobEntrIdentityNS}. The upper bound in Eq. \eqref{e:entropyBoundNS} then trivially follows from the Newton-Leibniz formula and positivity of integration. The lower bound is established through an analogous procedure to that presented in the proof of Theorem \ref{mainTh}.

\end{document}